\documentclass[a4paper,english, 11pt]{amsart}
\usepackage{amsthm}
\usepackage[foot]{amsaddr}
\usepackage{graphicx} \usepackage{array} \usepackage{colortbl,makecell}
\usepackage{pifont} 

\usepackage{amsmath, amssymb, amsfonts, verbatim}
\usepackage{hyphenat, epsfig, subcaption, multirow}
\usepackage{nicefrac}
\usepackage{paralist,enumitem}

\usepackage{dsfont} 

\usepackage{mathtools, etoolbox}
\usepackage{thmtools}
\usepackage{thm-restate}
\usepackage{xparse}

\usepackage[noend]{algorithmic}
\usepackage{algorithm}

\usepackage[font=small,labelfont=bf]{caption}

\usepackage[dvipsnames]{xcolor}

\usepackage{babel}
\usepackage{csquotes}

\usepackage[style=alphabetic,backref=true,maxnames=99,maxalphanames=99,
            url=false,isbn=false,doi=true]{biblatex}

\AtEveryBibitem{\clearname{editor}}
\AtEveryBibitem{\clearlist{publisher}}
\AtEveryBibitem{\clearfield{pages}}
\DeclareFontFamily{U}{mathx}{\hyphenchar\font45}
\DeclareFontShape{U}{mathx}{m}{n}{
      <5> <6> <7> <8> <9> <10>
      <10.95> <12> <14.4> <17.28> <20.74> <24.88>
      mathx10
      }{}
\DeclareSymbolFont{mathx}{U}{mathx}{m}{n}
\DeclareMathSymbol{\bigtimes}{1}{mathx}{"91}

\usepackage{tcolorbox}
\tcbuselibrary{skins,breakable}
\tcbset{enhanced jigsaw}

\definecolor{DarkRed}{rgb}{0.5,0.1,0.1}
\definecolor{DarkBlue}{rgb}{0.1,0.1,0.5}

\usepackage{nameref}
\definecolor{ForestGreen}{rgb}{0.1333,0.5451,0.1333}
\definecolor{Red}{rgb}{0.9,0,0}
\usepackage[linktocpage=true,
	colorlinks,
	linkcolor=DarkRed,citecolor=ForestGreen,
    bookmarks,bookmarksopen,bookmarksnumbered]
	{hyperref}
\usepackage[noabbrev,nameinlink,capitalize]{cleveref}
\crefname{property}{property}{Property}
\creflabelformat{property}{(#1)#2#3}
\crefname{equation}{Eq}{Eq}
\creflabelformat{equation}{(#1)#2#3}

\usepackage{bm}
\usepackage{url}
\usepackage{xspace}
\usepackage[mathscr]{euscript}
\usepackage{mathrsfs}

\usepackage{tikz}
\usetikzlibrary{arrows}
\usetikzlibrary{arrows.meta}
\usetikzlibrary{shapes}
\usetikzlibrary{backgrounds}
\usetikzlibrary{positioning}
\usetikzlibrary{decorations.markings}
\usetikzlibrary{patterns}
\usetikzlibrary{calc}
\usetikzlibrary{fit}
\tikzset{vertex/.style={circle, black, fill=Yellow, line width=1pt, draw, minimum width=8pt, minimum height=8pt, inner sep=0pt}}

\usepackage[framemethod=TikZ]{mdframed}

\usepackage[margin=1in]{geometry}

\usepackage{soul}

\renewcommand{\paragraph}[1]{\medskip\noindent\textbf{#1}}

\newtheorem{theorem}{Theorem}
\newtheorem{lemma}{Lemma}[section]
\newtheorem{proposition}[lemma]{Proposition}
\newtheorem{corollary}[lemma]{Corollary}
\newtheorem{claim}[lemma]{Claim}
\newtheorem{fact}[lemma]{Fact}

\newtheorem{definition}[lemma]{Definition}

\newtheorem*{claim*}{Claim}
\newtheorem*{proposition*}{Proposition}
\newtheorem*{lemma*}{Lemma}
\newtheorem*{problem*}{Problem}

\crefname{lemma}{Lemma}{Lemmas}
\crefname{claim}{Claim}{Claims}
\crefname{enumi}{Step}{Steps}
\crefname{step}{Step}{Step}

\newenvironment{proofclaim}[1][{\it Proof of claim. \hspace{0.066cm}}]{\noindent {}{#1}{}}{ \strut\hfill \qed \vspace{2ex}}

\newtheorem{remark}[lemma]{Remark}

\theoremstyle{definition}

\newenvironment{abox}{\begin{tcolorbox}[
		enlarge top by=5pt,
		enlarge bottom by=5pt,
frame hidden,
		overlay broken = {
			\draw[line width=1pt, black]
			(frame.north west) rectangle (frame.south east);},
		overlay = {
			\draw[line width=1pt, black]
			(frame.north west) rectangle (frame.south east);},
		boxsep=0pt,
		left=4pt,
		right=4pt,
		top=10pt,
		arc=0pt,
		boxrule=1pt,toprule=1pt,
		colback=white
	]}
{\end{tcolorbox}}

\newtheorem{mdalg}{Algorithm}
\newenvironment{Algorithm}{\begin{abox}\begin{mdalg}}{\end{mdalg}\end{abox}}

\renewcommand{\qed}{\nobreak \ifvmode \relax \else
      \ifdim\lastskip<1.5em \hskip-\lastskip
      \hskip1.5em plus0em minus0.5em \fi \nobreak
      \vrule height0.75em width0.5em depth0.25em\fi}

\renewcommand{\leq}{\leqslant}
\renewcommand{\geq}{\geqslant}
\renewcommand{\le}{\leqslant}
\renewcommand{\ge}{\geqslant}

\DeclarePairedDelimiter{\bracket}[]
\DeclarePairedDelimiter{\paren}()

\DeclarePairedDelimiter{\abs}{|}{|}

\DeclarePairedDelimiter{\ceil}{\lceil}{\rceil}

\DeclarePairedDelimiter{\set}{\{}{\}}

\DeclarePairedDelimiterXPP\lonenorm[1]{}\lVert\rVert{_1}{#1}

\DeclareMathOperator*{\distrib}{\mu}
\DeclareMathOperator*{\support}{supp}
\DeclareMathOperator*{\expect}{\boldsymbol{\mathrm{Exp}}}
\DeclareMathOperator*{\variance}{Var}

\DeclarePairedDelimiterXPP{\TVD}[2]{\Delta_{\textnormal{\text{TVD}}}}(){}{#1, #2}
\DeclarePairedDelimiterXPP{\PTD}[2]{\Lambda}(){}{#1, #2}

\DeclarePairedDelimiterXPP{\Ot}[1]{\widetilde{O}}(){}{#1}
\DeclarePairedDelimiterXPP{\Omgt}[1]{\widetilde{\Omega}}(){}{#1}
\DeclarePairedDelimiterXPP{\BigO}[1]{O}(){}{#1}

\NewDocumentCommand{\Prob}{sO{}E{_}{{}}m}{{\boldsymbol{\mathrm{Pr}}}_{#3}
  \IfBooleanTF{#1}
  {\bracket*{#4}}
  {\bracket[#2]{#4}}
}

\NewDocumentCommand{\Exp}{sO{}E{_}{{}}m}{\expect_{#3}
  \IfBooleanTF{#1}
  {\bracket*{#4}}
  {\bracket[#2]{#4}}
}

\DeclarePairedDelimiterXPP{\Var}[1]{\variance}[]{}{#1}
\DeclarePairedDelimiterXPP{\Cov}[1]{\covariance}[]{}{#1}
\DeclarePairedDelimiterXPP{\eexp}[1]{\exp}(){}{#1}

\NewDocumentCommand{\Dist}{sO{}E{_}{{}}m}{\distrib_{#3}
  \IfBooleanTF{#1}
  {\paren*{#4}}
  {\paren[#2]{#4}}
}

\DeclarePairedDelimiterXPP{\Supp}[1]{\support}(){}{#1}

\DeclarePairedDelimiterXPP{\KL}[2]{\kldiv}(){}{#1 \;\delimsize\|\; #2}

\DeclarePairedDelimiterXPP{\Ent}[1]{\entropy}(){}{#1}
\DeclarePairedDelimiterXPP{\Inf}[2]{\inform}(){}{#1 \; ; \; #2}

\renewcommand{\epsilon}{\varepsilon}
\newcommand{\eps}{\varepsilon}

\newcommand{\poly}{\mbox{\rm poly}}

\renewcommand{\wp}{$\text{w.p.}$\xspace}
\newcommand{\whp}{$\text{w.h.p.}$\xspace}
\newcommand{\wehp}{$\text{w.e.h.p.}$\xspace}

\newenvironment{tbox}{\begin{tcolorbox}[
		enlarge top by=5pt,
		enlarge bottom by=5pt,
		 breakable,
		 boxsep=0pt,
                  left=4pt,
                  right=4pt,
                  top=10pt,
                  arc=0pt,
                  boxrule=1pt,toprule=1pt,
                  colback=white
                  ]}
{\end{tcolorbox}}

\newcommand{\Language}[1]{\textnormal{\textsc{#1}}\xspace}

\newcommand{\Local}{\Language{Local}}
\newcommand{\LOCAL}{\Local}

\newcommand{\alg}[1]{\texttt{#1}\xspace}

\newcommand{\ov}[1]{\overline{#1}}

\newcommand{\rv}[1]{\boldsymbol{#1}}

\newcommand{\rL}{\rv{\mathrm{L}}}

\newcommand{\cS}{\mathcal{S}}

\newcommand{\cV}{\mathcal{V}}
\newcommand{\cU}{\mathcal{U}}
\newcommand{\cZ}{\mathcal{Z}}
\newcommand{\cX}{\mathcal{X}}
\newcommand{\cY}{\mathcal{Y}}

\newcommand{\fC}{\mathscr{C}}
\newcommand{\fL}{\mathscr{L}}
\newcommand{\fI}{\mathcal{I}}
\newcommand{\fR}{\mathscr{R}}
\newcommand{\fG}{\mathscr{G}}

\newcommand{\bN}{\mathbb{N}}

\newcommand{\bR}{\mathbb{R}}
\newcommand{\bF}{\mathbb{F}}

\newcommand{\col}{\varphi}

\DeclareMathOperator*{\dom}{dom}

\newcommand{\Vsparse}{V_{sparse}}

\title{Beyond Brooks: $(\Delta-1)$-Coloring in Semi-Streaming} 

\author{Maxime Flin}
\address{Aalto University, Finland.}
\email{maxime.flin@aalto.fi}
\author{Magnús M. Halldórsson}
\address{Reykjavik University, Iceland.}
\email{mmh@ru.is}

\thanks{M.\ Flin was supported in part by the Research Council of Finland, Grants 359104 and 363558. Part of this work was done while the first author was working at Reykjavik University, funded by the Icelandic Research Fund, Grant 2310015-053. M.M.\ Halld\'orsson was supported by the Icelandic Research Fund, Grant 2511609.}
\date{}

\begin{document}
\begin{abstract}
Reed [J.~Comb.~Theory B, 1999] showed that graphs of maximum degree $\Delta \geq 10^{14}$ without $\Delta$-cliques are $(\Delta-1)$-colorable. 
We design a one-pass semi-streaming algorithm for computing such a coloring. Additionally, we prove that any one-pass $(\Delta-k)$-coloring algorithm for $0\leq k < (\Delta+1)/2$ requires $\Omega(n(k+1))$ space.
\end{abstract}
 \maketitle

\section{Introduction}

Graph coloring is a central problem in combinatorics and theoretical computer science. Given a graph $G=(V,E)$ and an integer $q \ge 1$, a $q$-coloring assigns a color from $\set{1, \dots, q}$ to each vertex so that adjacent vertices receive different colors.
Coloring has been extensively studied in diverse computational models, including the classical RAM \cite{LundY94, KhannaLS00}, distributed message-passing \cite{Linial92,BEPS16,HSS18,CLP20}, sublinear models \cite{ACK19,CFGUZ19,CDP21,AY25}, dynamic algorithms \cite{BCHN18,HP22,BGKL22,BRW24}, communication complexity \cite{FM25,CMNS25}, 
and streaming models \cite{ACK19,AA20,BBMU21,ACS22,CGS22,AKM23,AY25}.
In semi-streaming, the algorithm must produce a coloring after processing the edges of an $n$-vertex graph one at a time in an adversarial order while using $O( n \cdot \poly\log n )$ space.

In this work, we study the streaming complexity of coloring with fewer than $\Delta$ colors, where $\Delta$ is the maximum degree.
We give the first one-pass semi-streaming algorithm that outputs a $(\Delta-1)$-coloring for graphs with large enough degree and no $\Delta$-clique. 

\paragraph{Coloring With $\Delta+1$ or $\Delta$ Colors.}
The classic $(\Delta+1)$-coloring problem is easy sequentially because one can always \emph{extend a partial solution} to a full solution. 
Notwithstanding, the greedy approach does not work in semi-streaming; thus to $(\Delta+1)$-color in semi-streaming, Assadi, Chen, and Khanna \cite{ACK19} had to invent the influential \emph{palette sparsification} technique: if one samples $O(\log n)$ colors from $\set{1, 2, \ldots, \Delta+1}$ independently for each vertex, then with high probability, all vertices can be colored with a color from their list. It has been further leveraged in sublinear algorithms \cite{AA20,AKM23,AY25}, distributed algorithms \cite{HKNT22,FGHKN24}, dynamic algorithms \cite{BRW24}, and discrete mathematics \cite{KK23,KK24,Dhawan24}.

By a theorem of Brooks \cite{brooks1941}, $\Delta$ colors always suffice except for $(\Delta+1)$-cliques or odd cycles.
Assadi, Kumar, and Mittal \cite{AKM23} gave a semi-streaming algorithm for $\Delta$-coloring, in spite of the fact that palette sparsification provably fails \cite[Proposition C.2]{ACK19} and such an algorithm only exists when each edge appears at most once in the stream \cite{AKM23}. To that end, they introduced the novel sketching technique of sparse recovery.

In the scale of things, $\Delta$-coloring is conceptually easy. It has spawned a great many solutions, starting with \cite{MV69,LOVASZ75}, including newer
linear-time algorithms \cite{BW14,SS25};
see comprehensive treatments in \cite{CR13} and the book of \cite{ST15}. 
\emph{Modulo a single vertex}, 
it can be greedily colored, 
as there is an order under which it suffices to assign each node an arbitrary available color.
The $(\Delta-1)$-coloring problem does not have such a greedy or near-greedy property.

\paragraph{Beyond Brooks: $(\Delta-1)$-Coloring.}
In \cite{Reed98}, Reed showed that Brooks barely scratched the surface: in fact, as the size of the largest clique decreases, so does the chromatic number.
Borodin and Kostochka \cite{BK77} conjectured that every graph of maximum degree $\Delta \ge 9$ without a $\Delta$-clique is $(\Delta-1)$-colorable.
Reed \cite{Reed99} proved the conjecture for $\Delta \ge 10^{14}$. 
The conjecture is otherwise open, but holds for claw-free graphs \cite{CR13}. 
Despite significant advances in $(\Delta+1)$- and even $\Delta$-coloring, going further in the semi-streaming model has remained elusive.

Unsurprisingly, Reed's result is an order of magnitude more involved than the many simple proofs of Brooks' theorem. Reed's approach --- based on technical structural observations and the probabilistic method --- is highly non-trivial and remains (effectively) the \emph{only one known}.
As such, Reed's theorem is not merely a quantitative improvement but an exhibit for the existence of a conceptual leap between $(\Delta-1)$-coloring and $\Delta$-coloring.
Indeed, unlike proofs of Brooks' Theorem (including that of \cite{AKM23}), Reed's argument is far from being greedy: it begins by transforming the graph --- specifically, by adding edges.
Given the already strenuous effort necessary to obtain a $(\Delta+1)$-coloring and $\Delta$-coloring in semi-streaming and the technical depth of Reed's theorem, it is natural to ask if $\Delta$-coloring is the best we can hope for in semi-streaming:
\begin{quote}
    \emph{Does there exist a (one-pass) semi-streaming algorithm to color with fewer than $\Delta$ colors any given graph of maximum degree $\Delta$ (large enough) and no $\Delta$-cliques?}
\end{quote}
Even if such an algorithm exists, the concern is that removing one more color may require an explosion of cases, considering, e.g., the much larger number of cases in the $\Delta$-coloring algorithms \cite{AKM23,FHM23} than for $(\Delta+1)$-coloring \cite{ACK19,HNT21}.

\subsection{Our Contributions}
The main technical contribution of this paper is a semi-streaming algorithm for $(\Delta-1)$-coloring graphs of sufficiently high degree. Concretely, we prove the following.

\begin{theorem}
    \label{thm:streaming}
    There exists a universal constant $\Delta_0$ for which the following holds.
    There exists a randomized one-pass semi-streaming algorithm that given any graph $G=(V,E)$ with maximum degree $\Delta \geq \Delta_0$ and containing no $\Delta$-clique, outputs a $(\Delta-1)$-coloring of $G$ with high probability.
\end{theorem}

During the streaming pass, akin to \cite{AKM23}, our algorithm samples $\poly(\log n)$-sized random lists and stores only the edges of the sparsified graph, meanwhile using their sparse recovery technique to learn about the densest regions of the graph.
The coloring algorithm differs substantially: we begin by transforming the graph --- both removing extremely dense subgraphs and carefully adding edges --- and then leverage various kinds of probabilistic coloring arguments to construct a coloring.

Clearly, the impossibility results highlighted in \cite{AKM23} also apply to us. We therefore assume that each edge appears at most once in the stream.
It is also worth noting that randomness is essential, as Assadi, Chen, and Sun \cite{ACS22} demonstrated that any deterministic semi-streaming algorithm requires $\exp(\Delta^{o(1)})$ colors.

As a final remark on \cref{thm:streaming}, we observe that $\Delta_0$ is in fact the smallest universal constant for which graphs of maximum degree $\Delta \geq \Delta_0$ and no $\Delta$-clique admit a $(\Delta-1)$-coloring. Reed proved in \cite{Reed99} that $\Delta_0 \leq 10^{14}$. For graphs of maximum degree $\Delta \leq 10^{14}$, all edges $O(n\Delta) = O(n)$ in the graph can be stored, and an optimal coloring can be computed (although, not in polynomial time).
We have not tried to reduce this constant; in \cite[Section 5]{Reed99}, Reed claims that $\Delta_0 = 10^6$ also works with more care, while $\Delta_0 < 100$ would need different techniques.

\paragraph{Using Fewer Colors.}
A natural follow-up question from \cref{thm:streaming} is to what extent this can be pushed further.
Molloy and Reed \cite{MR14} gave a characterization of $(\Delta-k)$-colorable graphs in terms of forbidden local subgraphs, for all $k \lesssim \sqrt{\Delta}$. This generalizes Reed's result and suggests an alternative approach. However, their structural decomposition is delicate and obtained through phases of contracting vertices and computing optimal partitions. We prove that, in fact, computing a $(\Delta - k)$-coloring requires $\Omega( n (k + 1) )$ space, prohibiting semi-streaming algorithms for $(\Delta - k)$-coloring unless $k$ is polylogarithmic.

\begin{restatable}{theorem}{LowerBound}
    \label{thm:lowerbound}
    Let $\Delta, c > 0$ be non-negative integers such that $(\Delta + 1)/2 < c \leq \Delta$.
    Any (possibly randomized) one-pass streaming algorithm that outputs a $c$-coloring of given $c$-colorable $n$-vertex graph of maximum degree $\Delta$ requires $\Omega( n (\Delta - c + 1) )$ space.
\end{restatable}

Very recently, Assadi, Sundaresan, and Yazdanyar \cite{ASY26} studied the problem of distinguishing graphs of low chromatic number (e.g., constant) from those with a large chromatic number (e.g., $\sqrt{n}$). They proved, amongst other things, that it requires $n^{2 - o(1)}$ space in adversarial streams. \cref{thm:lowerbound} is orthogonal to (and much simpler than) their results since it considers graphs whose chromatic number is one less than their clique number.

\subsection{Organization of the Paper.}
In \cref{sec:tech-intro}, we give an overview of the challenges and technical ideas necessary to obtain \cref{thm:streaming}. At the end of \cref{sec:tech-intro}, we also provide a more detailed comparison with \cite{AKM23,anonymous}. In \cref{sec:streaming-alg}, we detail the information stored during the streaming pass and describe the coloring algorithm in \cref{sec:coloring-alg}. \cref{thm:lowerbound} is presented in \cref{sec:lowerbound}.

 \section{Technical Introduction}
\label{sec:tech-intro}

We give an overview of the ideas behind \cref{thm:streaming}.
First, we present the three components for constructing $(\Delta-1)$-colorings, followed by a description of the high-level algorithm that combines them. We then explain how to implement these components in semi-streaming.
For a complete step-by-step description, see \cref{sec:coloring-alg}.
In this introduction, various concepts are simplified to aid intuition.

\paragraph{Pre-/Post-Processing.}
We start by identifying subgraphs in the graph that are easily colorable regardless of the coloring in the rest of the graph. During pre-processing, we remove those subgraphs and only color them at the very end during post-processing.

For instance, consider the $(\Delta+1)$-clique missing two disjoint edges, as illustrated in \cref{fig:2-anti-matching}. In any $(\Delta-1)$-coloring of such a subgraph, both endpoints of each anti-edge must be assigned the same color. Once these pairs of vertices are colored the same, the remaining vertices can be colored greedily. In general, the subgraphs shown in \cref{fig:degree-1-choosable} are such that any coloring of the outer vertices can be extended to a coloring of the inner vertices as well (they are, in fact, $(\deg-1)$-choosable; see \cref{sec:choosable-subgraphs}). This property allows us to defer their coloring to the very end of the algorithm.

To detect these subgraphs, during pre-processing, we decompose the graph into \emph{(locally) sparse vertices} --- which have many missing edges in their neighborhood --- and \emph{almost-cliques} --- sets of vertices that can be turned into $(\Delta+1)$-cliques by modifying an $\eps$-fraction of their edges and vertices. We look for those choosable subgraphs only in the densest almost-cliques, as the other vertices can be colored by other means.
We then distinguish two main types of choosable components:
the \emph{solitary} ones are inscribed inside an almost-clique (as on \cref{fig:2-anti-matching,fig:3-IS}), and the \emph{popular} ones for which the choosable component contains some vertices outside the almost-clique (as on \cref{fig:populaire}).

\begin{figure}[ht!]
    \centering
    \begin{subfigure}[b]{.3\linewidth}
        \centering
        \includegraphics[width=\linewidth,page=1]{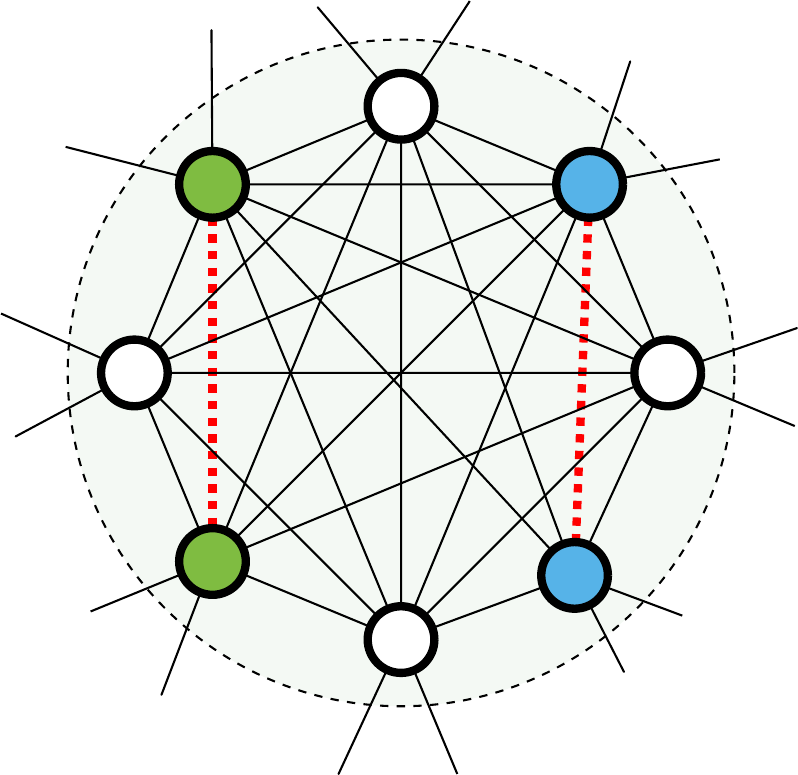}
        \caption{\label{fig:2-anti-matching}}
    \end{subfigure}
    \hfill
    \begin{subfigure}[b]{.3\linewidth}
        \centering
        \includegraphics[width=\linewidth,page=2]{figures/DCCs.pdf}
        \caption{\label{fig:3-IS}}
    \end{subfigure}
    \hfill
    \begin{subfigure}[b]{.3\linewidth}
        \centering
        \includegraphics[width=\linewidth]{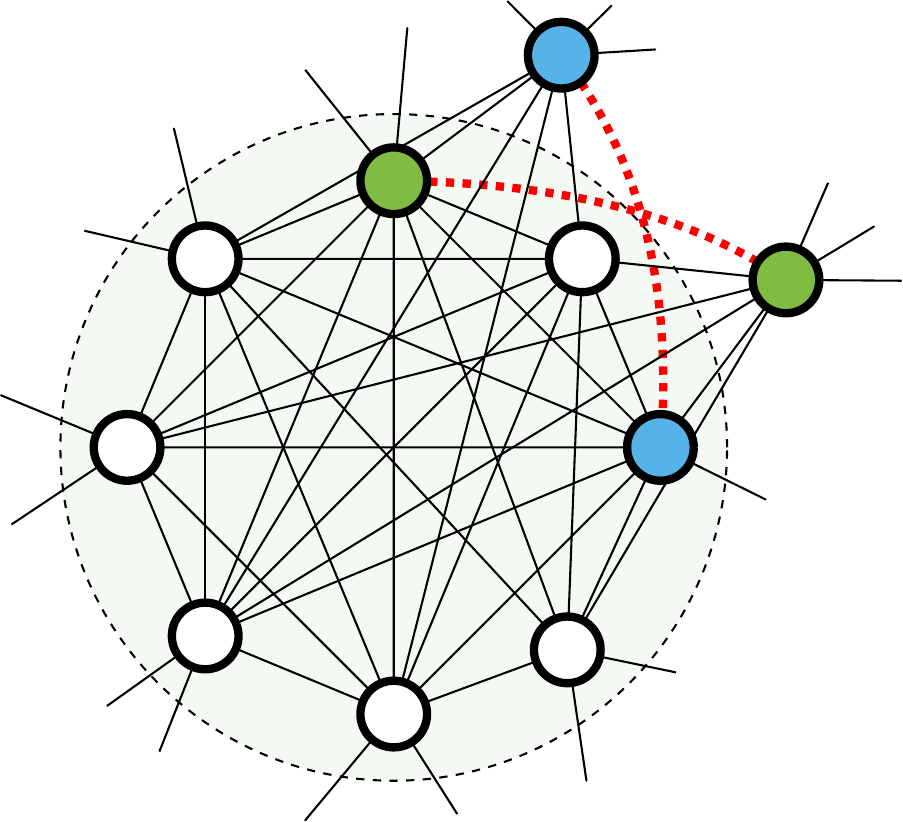}
        \caption{\label{fig:populaire}}
    \end{subfigure}

    \caption{The three types of subgraphs $G[X]$ such that any coloring of $G[V - X]$ can be extended to a coloring of $G$. The vertices inside the gray circle are in the almost-clique. \cref{fig:2-anti-matching} represent an almost-clique with an induced 2-anti-matching. \cref{fig:3-IS} represents an almost-clique with a 3-independent set. \cref{fig:populaire} represents a $(\Delta-1)$-clique with two vertices sharing a neighbor in the clique, each with a different anti-neighbor in the clique.
    To color each such graph, we must ensure that the highlighted pairs of vertices are same-colored, which is unlikely to occur through randomized coloring arguments.
    \label{fig:degree-1-choosable}}
\end{figure}

For the remainder of this technical introduction, we focus on the case of $\Delta$-regular nodes in disjoint $(\Delta-1)$-cliques, as this already captures the main challenges and intuition for $(\Delta-1)$-coloring.

\paragraph{Generating Slack Probabilistically.}
Rather than directly attempting to color all vertices, many coloring algorithms begin by coloring a few vertices randomly. The hope is that uncolored vertices have neighbors that pick the same color, which increases the flexibility when completing the coloring.

Consider the following random coloring process.
Each vertex flips a coin and, if it is heads, randomly picks a color. A vertex keeps its color only if no neighbor selects the same one.
If $v$ is a vertex in a $(\Delta-1)$-clique such as the one in \cref{fig:triplets}, basic probabilistic analysis shows that $v$ has a constant probability of saving two colors --- meaning that $v$ remains inactive, and both of its external neighbors retain colors used elsewhere in the clique. We say that $v$ was successful.

For a $(\Delta - 1)$-clique with all (or most) of its external neighbors distinct, one can show that with high probability it contains two successful vertices. Then, we can defer the coloring of successful vertices and color the rest of the clique greedily.
This technique generalizes to any almost-clique without an external node connected to most of its internal vertices.

\paragraph{Reed Transform.}
The most difficult case occurs when a $(\Delta-1)$-clique $C$ has an external vertex --- called a friend --- with numerous neighbors in $C$, yet lacks any structure that yields probabilistic slack or choosability. To illustrate the challenge, consider the situation where the vertices are adjacent to either $s_1$ or $s_2$, but not both. If the probabilistic coloring procedure colors $s_1$ and $s_2$ the same (an event with probability $\approx 1/\Delta$), then the coloring cannot be extended to the rest of $C$ (see \cref{fig:amicale}). It is crucial to note that recoloring vertices might convert successful vertices into \emph{un}successful ones, rendering the approach infeasible.

To handle such cliques, we apply the Reed Transform, based on a construction from \cite{Reed99}. Identify a specific substructure in $C$ (vertices $s, x, y, u, v$ as in \cref{fig:reed-transform}), delete $C$ from the graph, and add a single edge $xy$ (in dashed gray in \cref{fig:reed-transform}). Once the rest of the graph is colored --- forcing $x$ and $y$ to have distinct colors --- we reinsert and color $C$ by coloring $v$ as $x$ and $s$ as some anti-neighbor $z \in C - N(s)$.
Naturally, $x$ and $y$ must be chosen carefully to avoid creating a $\Delta$-clique when we connect them.

\begin{figure}[ht!]
    \centering
    \begin{subfigure}{.33\linewidth}
        \centering
        \includegraphics[width=\linewidth,page=1]{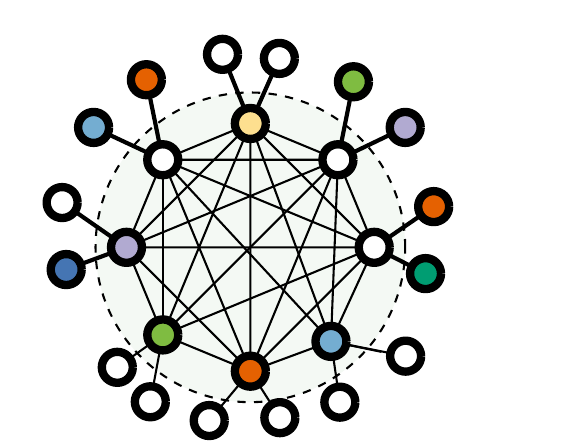}
        \caption{\label{fig:triplets}}
    \end{subfigure}
    \hfill
    \begin{subfigure}{.33\linewidth}
        \centering
        \includegraphics[width=\linewidth,page=2]{figures/triples.pdf}
        \caption{\label{fig:amicale}}
    \end{subfigure}
    \hfill
    \begin{subfigure}{.3\linewidth}
        \centering
        \includegraphics[width=\linewidth,page=3]{figures/triples.pdf}
        \caption{\label{fig:reed-transform}}
    \end{subfigure}
    \caption{A $(\Delta-1)$-clique in different configurations. On \cref{fig:triplets} it has many different external neighbors. On \cref{fig:amicale}, all vertices are either adjacent to $s_1$ on the left or $s_2$ on the right. If slack generation colors $s_1$ and $s_2$ the same, there is no way to extend the coloring. On \cref{fig:reed-transform}, a representation of the Reed Transform: the clique is removed, and the dashed edge added to the graph; the coloring can be later extended by same coloring the $s$ vertex with some vertex of the clique and the $v$-vertex like the $x$-vertex. The edge added between the $x$- and $y$-vertex ensures this is possible.}
\end{figure}

\paragraph{The Coloring Algorithm.}
With this, we now describe the high-level structure of our $(\Delta-1)$-coloring algorithm.
Despite our intricate analysis, each step is conceptually simple:
\begin{enumerate}
    \item Identify and remove choosable components (solitary and popular as in \cref{fig:degree-1-choosable});
    \item Apply the Reed Transform to troublesome $(\Delta-1)$-cliques;
    \item Use probabilistic slack to color the remaining graph greedily;
    \item Invert the Reed transform to color the removed cliques;
    \item Complete the coloring of the choosable components.
\end{enumerate}
Note how the steps pair with each other and correspond to one of the aforementioned three ideas.
For the remainder of this technical introduction, we outline how each step can be implemented in semi-streaming, discuss the challenges, and present our solutions.

\paragraph{Streaming Implementation.}
We focus here on the implementation of our coloring algorithm after the streaming pass is done, for the streaming algorithm itself largely follows from \cite{ACK19,AKM23}: running palette sparsification and computing linear sketches. See \cref{sec:summary-streaming} for a summary of the information recovered by the algorithm after the streaming pass.

First, we observe that we may focus on the densest almost-cliques, for the other ones can be colored by palette sparsification. Using the sparse-recovery technique of \cite{AKM23}, we recover almost all the edges incident to the densest almost-cliques. In particular, it allows us to identify the choosable subgraphs and the $(\Delta-1)$-cliques that need to be removed by the Reed Transform.
Implementing the Reed Transform requires that we avoid the creation of a $\Delta$-clique and that we preserve invertibility for later steps. We must also ensure that we do not densify sparse neighborhoods, as this could inadvertently create the dense subgraphs we sought to isolate in Step (1).
Remarkably, simple downsampling suffices: it limits the number of edges added per vertex while preserving enough flexibility to avoid creating a $\Delta$-clique.

We avoid monochromatic edges by keeping a \emph{serene coloring}: a coloring for which each node uses a color from its random list unless all its incident edges have been recovered, in which case it may choose its color freely.
We frequently combine palette sparsification with sparse-recovery. For instance, when coloring popular almost-cliques, some nodes are colored using their random lists, while others are colored greedily, as in the classical setting. Overall, our algorithm makes significantly more
use of the edges recovered by sparse recovery than \cite{AKM23}.
Though subtle, this difference proves consequential: it frees us to focus more deeply on the graph’s structural properties.

The second major challenge is a technical one: proving concentration on the probabilistic guarantees of the randomized coloring procedure described earlier.
Indeed, to invert the Reed Transform, we require that each friend has at least $\Delta/\poly(\log n)$ edges into the clique. Consequently, we can no longer guarantee that most external neighbors are distinct. At a high level, this undermines the apparent randomness of colors sampled outside the almost-clique—they may originate from too few vertices to provide sufficient diversity. Unfortunately, standard concentration inequalities fall short in capturing this behavior. Instead, we rely on a case analysis that combines multiple forms of slack-based concentration, sometimes even blending probabilistic and temporal slack to recover the guarantees we need.

\paragraph{Comparison with \cite{AKM23}.}
In \cite{AKM23}, the authors design a semi-streaming algorithm for $\Delta$-coloring and, in particular, introduces the sparse-recovery technique used by our algorithm. In contrast, while \cite{AKM23} would use it to recover edges incident on at most two vertices for each of the densest cliques (see Definitions 4.6 and 4.7 of \cite{AKM23}), we heavily use that we can recover all edges incident on the densest cliques, most notably for implementing and inverting the Reed Transform.

\paragraph{Coloring with fewer than $(\Delta-1)$-colors.}
Molloy and Reed \cite{MR14} generalized Reed's result to $(\Delta - k)$-colorable graphs, for all values of $k$ up to roughly $\sqrt{\Delta}$. As shown in \cite{BE19,anonymous}, this can be obtained in the LOCAL model \cite{Linial92,HS20} in a near-optimal round complexity of $\poly(\log\log n)$. The lower bound of \cref{thm:lowerbound} proves that $(\Delta - k)$-coloring is not attainable in semi-streaming unless $k = \poly(\log n)$.

Our work leaves open the question of $(\Delta - k)$-coloring in the streaming model for $k \ge 2$, specifically whether it can be achieved in $n \cdot \poly(k, \log n)$ space. Conceptually, as $k$ increases, the number of cases involving dense cliques grows rapidly, as 
it no longer suffices to track two anti-edges per dense clique.

To avoid this combinatorial explosion,  the main alternative is to follow the approach of \cite{MR14}, which provides a systematic treatment for arbitrary $k$ via a generalization of the Reed transform. This framework starts from a coarse structural decomposition (akin to the almost-clique decomposition) and progressively refines it through increasingly intricate local transformations. These transformations appear to require the full knowledge of the local topology as they rely, among other things, on the optimal coloring of dense induced subgraphs. While it is conceivable that, using the sparse recovery technique with $\poly(k)$-sparse vectors (thus, $n \cdot \poly(k,\log n)$ space), one could recover most of the edges necessary to implement this intricate transformation, the feasibility of this approach remains open.
 
\section{Preliminaries}

\paragraph{Notation.} 
For an integer $k \geq 1$, we use $[k]$ to denote $\set{1, 2, \ldots, k}$. For sets $A$ and $B$, we write $A - B = \set{a \in A : a \notin B}$. For succinctness, we abuse notation slightly and write $A - x$ and $A + x$ for $A - \set{x}$ and $A \cup \set{x}$ respectively. For a function $f: X \to Y$, let $f(A) = \set{f(a): a\in A}$ and $f^{-1}(B) = \set{x : f(x) \in B}$ where $A \subseteq X$ and $B \subseteq Y$ respectively. We abuse notation slightly and occasionally write $f^{-1}(y)$ instead of $f^{-1}(\set{y})$ when $y \in Y$.

Let $G=(V,E)$ be a simple graph. We denote by $n=|V|$ and $m=|E|$ its respective number of vertices and edges. For a set $S\subseteq V$, the induced subgraph $G[S]$ is the graph on vertex set $S$ and all edges of $E$ between pairs of vertices in $S$. The neighborhood of $v\in V$ is $N_G(v)=\set{u\in V: \set{u,v}\in E}$ and $\deg(v,G) = |N_G(v)|$. The closed neighborhood of $v$ is $N_G[v] = N_G(v) + v$. When $G$ is clear from context, we write $N(v)$ and $\deg(v)$. The maximum degree of $G$ is denoted by $\Delta$. We refer to a pair $\set{u,v}$ as a non-edge or an anti-edge when no edge connects $u$ and $v$. A $k$-anti-matching is a set of $k$ anti-edges such that each vertex belongs to at most one anti-edge. A $k$-independent set is a set of $k$ vertices that are not connected by any edge.

For an integer $q \geq 1$, a partial $q$-coloring is a mapping $\col : V \to [q] \cup \set{\bot}$ such that $\col(u) \neq \col(v)$ or $\bot \in \set{\col(u), \col(v)}$ for all $\set{u,v}\in E$. The set of colored vertices is the domain $\dom\col = \set{v\in V: \col(v) \neq\bot}$. Given a partial coloring $\col$, the uncolored degree of $v$ is $\deg_\col(v, G) = |N_G(v) - \dom\col|$. Throughout the paper, we denote by $L_\col(v) = [\Delta-1] - \col(N(v))$ the set of colors unused by neighbors of $v$ with respect to a (possibly partial) coloring $\col$. We say that a coloring $\psi$ is an extension of $\col$ if $\psi(v) = \col(v)$ for all $v\in \dom\col$.

We say that an event occurs \emph{with high probability} (in $n$) if it occurs with probability at least $1 - n^{-c}$, for a desirably large constant $c > 0$. An event occurs \emph{with exponentially high probability} in $x$ if it occurs with probability at least $1 - \exp(-\Omega(x))$. We abridge ``with probability'', ``with high probability'' and ``with exponentially high probability'' as ``\wp'', ``\whp'' and ``\wehp'' respectively.

\subsection{Sparse-Dense Decomposition}
We use the variant of the sparse-dense decomposition from \cite{HSS18,ACK19,AW22,AKM23} (inspired by \cite{Reed98,RM02}). It partitions vertices between (locally) sparse vertices --- with many anti-edges in their neighborhood --- and dense vertices clustered in almost-cliques --- subgraphs that resemble a $(\Delta+1)$-clique.

\begin{definition}
A vertex $v$ is \emph{\underline{$\zeta$-sparse}} in $G$ if $G[N(v)]$ contains at most $\binom{\Delta}{2} - \zeta \Delta$ edges.
\label{def:sparse}
\end{definition}

\begin{definition}
    \label{def:AC}
    A set $C \subseteq V$ is an \emph{\underline{$(\eps,\delta)$-almost-clique}} if 
    \begin{enumerate}
        \item\label[part]{part:acd-size}
            $(1 - \eps/2)\Delta \leq |C| \leq (1 + \eps/2)\Delta$,
        \item\label[part]{part:acd-neighbors}
            $|C - N[v]|, |N(v) - C| \leq \eps\Delta$ for all $v \in C$, and
        \item \label[part]{part:acd-ext}
            every $v \notin C$ has $|C - N(v)| \ge \delta \Delta$.
    \end{enumerate}
\end{definition}

We stress that (\ref{part:acd-ext}) is not always included as part of the definition of almost-clique (e.g., in \cite{HSS18,ACK19}). It is however essential for us that nodes outside an almost-clique have $\delta\Delta$ non-neighbor inside the almost-clique (we will always have $\delta = \Omega( \eps )$).

Let $v$ be a vertex in some almost-clique $C$. We call external neighbors $E(v) = N(v) - C$ its $e(v) = |E(v)|$ neighbors outside $C$. The vertices of $C$ that are not adjacent to $v$ are its anti-neighbors $A(v) = C - N[v]$; let $a(v) = |A(v)|$ be the anti-degree of $v$. Part (\ref{part:acd-neighbors}) in \cref{def:AC} means that $e(v), a(v) \leq \eps\Delta$.

\begin{definition}
\label{def:acd}
For $\eps,\eta \in (0,1)$, an \emph{\underline{$(\eta,\eps,\delta)$-almost-clique decomposition}} of $G$ is a vertex-partition $\Vsparse, C_1, C_2, \ldots, C_k$ such that
\begin{enumerate}
    \item every $v\in \Vsparse$ is $\eta \eps^2 \Delta$-sparse;
        \label[part]{part:acd-sparse-nodes}
    \item for every $i\in [k]$, the set $C_i$ is an $(\eps,\delta)$-almost-clique.
\end{enumerate}
\end{definition}

Vertices of $\Vsparse$ are called sparse, and those of $V - \Vsparse$ are called dense. It follows almost directly from the definitions that dense vertices have sparsity proportional to their anti-degree and external-degree. For a proof, we refer readers to \cite[Lemma 6.2]{HKMT21} for the bound in terms of anti-degrees and to \cite[Lemma 3.6]{FH25} for the bound in terms of external degrees\footnote{In contrast to our definition, which includes sparse vertices as external neighbors, in \cite{HKMT21}, the authors count as external neighbors only vertices from other almost-cliques. The original
proof of \cref{fact:acd-dense-sparsity} is in 
a technical report 
\cite{HNT21}.}.

\begin{lemma}
\label{fact:acd-dense-sparsity}
If $\Vsparse, C_1, \ldots, C_k$ is an $(\eta,\eps,\delta)$-almost-clique decomposition, then every $v\in C_i$ with $e(v) \geq 2/(\eta\eps^2)$ is $\max\set{ \eta\eps^2/4 \cdot e(v), (1-3\eps)/2 \cdot a(v)}$-sparse. 
\end{lemma}

A semi-streaming algorithm for computing an almost-clique decomposition exists, as given by {\cite[Proposition 3.5]{AKM23}}.

\begin{proposition}
    \label{lem:compute-acd-streaming}
    There are universal constants $\eta, \eps_0 \in (0,1)$ such that the following holds. For any $0 < \eps < \eps_0$, there is a streaming algorithm using $O\paren{ \eps^{-2}n\log n }$ space
    that computes an $(\eta,\eps,\eps)$-almost-clique decomposition.
\end{proposition}

\subsection{Palette Sparsification}
Our streaming algorithm uses the palette sparsification technique from \cite{ACK19} to color some almost-cliques. Assadi, Chen and Khanna view the coloring of an almost-clique as a perfect matching in a bipartite graph, where one side of the bipartition represents vertices and the other side the colors. An edge exists between a vertex and a color in this graph if the color is available to the vertex with respect to the current coloring. When we apply palette sparsification, each edge is retained in the graph independently with some probability, and they prove that, under the right assumptions, it is possible to match all vertices with high probability.

We use the following formulation of the statement, from \cite[Lemma 3.10]{AKM23}, in a black-box manner. We refer readers to \cite[Section 3.5]{ACK19} for a detailed proof. We say that a matching is an $\fL$-perfect matching iff it matches all vertices of $\fL$.

\begin{lemma}
    \label{lem:palette-graph}
    Let $\fG = (\fL \cup \fR, E)$ be a bipartite graph with the following properties:
    \begin{enumerate}
        \item $k = |\fL|$ and $k \leq |\fR| \leq 2k$;
        \item every vertex $v\in \fL$ has $\deg(v, \fR) \geq 2k/3$; and
        \item for every set $S \subseteq \fL$ with $|S| \geq k/2$ we have that $\sum_{v\in S} \deg(v, \fR) \geq |S| \cdot k - k/4$.
    \end{enumerate}
    For any $\delta \in (0,1)$, the subgraph obtained by sampling each edge in $\fG$ independently with probability at least $\frac{20}{k}\paren*{ \log k + \log(1/\delta)}$ contains an $\fL$-perfect matching with probability at least $1 - \delta$.
\end{lemma}

As observed by \cite{AKM23}, every almost-clique $C$ such that $G[C]$ contains at least $10^7 \eps\Delta$ anti-edges (they call them holey) can be colored using palette sparsification. Note that the extension exists with high probability regardless of the coloring outside $C$. 
\begin{lemma}
    \label{lem:streaming-color-trouée}
    Let $C$ be an almost-clique inducing at least $10^7\eps\Delta$ anti-edges and let $\rL(v)$ be a list of $O_\eps(\log n)$ random colors in $[\Delta-1]$ for each $v\in C$. For any partial $(\Delta-1)$-coloring $\col$ outside $C$ (i.e., $C \cap \dom\col = \emptyset$), there exists, \whp, an extension $\psi$ of the coloring to $C$ such that $\psi(v) \in \rL(v)$ for every $v\in C$.
\end{lemma}

We stress that \cref{lem:streaming-color-trouée} assumes that $C$ is uncolored.
Since the proof is almost verbatim that of \cite[Lemma 5.10]{AKM23} or \cite{ACK19}, we skip details and refer interested readers to \cite[Appendix B.3]{AKM23}.

\subsection{Concentration}
We use the classic Chernoff Bound on sums of independent random variables:
\begin{proposition}[{Chernoff bounds, \cite[Theorem 1.10.1 and 1.10.5]{Doerr2020}}]\label{lem:chernoff}
Let $X_1, \ldots, X_n$ be a family of independent random variables with values in $[0,1]$, and let $X = \sum_{i \in [n]} X_i$.
Suppose $\mu_{\mathsf{L}} \leq \Exp{X} \leq \mu_{\mathsf{H}}$, then
\begin{align}
  \text{for all } \delta \geq 0, \text{ we have that }
  &&
  \Prob*{ X > (1+\delta) \mu_{\mathsf{H}} }
  & \le \exp\paren*{-\frac{\delta^2}{2 + \delta}\mu_{\mathsf{H}}}\ .
  \label{chernoff-mult-upperbound}
\\
  \text{for all } \delta \in(0,1), \text{ we have that }
  &&
  \Prob*{ X < (1-\delta) \mu_{\mathsf{L}} }
  & \le \exp\paren*{-\frac{\delta^2}{2}\mu_{\mathsf{L}}}\ .
  \label{chernoff-mult-lowerbound}
\end{align}
\end{proposition}

Let $Y_1, \ldots, Y_n$ be boolean random variables (with values in $\set{0,1}$). We say they form a \emph{read-$k$ family} if they can be expressed as a function of independent random variables $X_1, \ldots, X_m$ such that each $X_j$ influences at most $k$ variables $Y_i$. More formally, there are sets $P_i \subseteq [m]$ for each $i \in [n]$ such that: (1) for each $i\in[n]$, the variable $Y_i$ is a function of $\set{X_j: j \in P_i}$ and (2) $|\set{i: j\in P_i}| \le k$ for each $j \in [m]$.

\begin{proposition}[read-$k$ bound, \cite{GLSS15}]
\label{lem:k-read}
Let $Y_1, \ldots, Y_n$ be a read-$k$ family of boolean variables, and let $Y$ be their sum.
Then, for any $\delta > 0$,
\[
  \Prob*{ \abs*{Y - \Exp{Y}} > \delta n}
  \le 2\exp\paren*{-\frac{2\delta^2 n}{k}} \ .
\]
\end{proposition}

 \section{The Streaming Algorithm}
\label{sec:streaming-alg}

In this section, we describe the information we collect during the streaming pass. This amounts to the almost-clique decomposition (\cref{lem:compute-acd-streaming}), the sparsified graph (\cref{sec:streaming-palette-sampling}) and some linear sketches to recover edges in the densest regions of the graph (\cref{sec:streaming-sparse-recovery}). Overall, we use $\Ot{n}$ space with high probability. In \cref{sec:structural-decomposition}, we introduce structural definitions essential to our coloring algorithm.

\paragraph{Parameters.}
Let us define some parameters:
\begin{eqnarray*}
    &\alpha = 150 \quad&\text{a conveniently large constant,}\\
    &\eps = 10^{-8} \quad&\text{the parameter of the almost-clique decomposition in \cref{lem:compute-acd-streaming,}}\\
    &\rho = \frac{c\log n}{\eps^2} \quad&     \text{{the sparsity threshold above which palette sparsification works,}}
\end{eqnarray*}
where the constant $c$ is large enough to have $\exp(-\Omega(\eps^2\rho)) \leq 1/\poly(n)$. We use the dependency on $\eps$ of $\rho$ in \cref{sec:coloring-H-not-critical}, \cref{eq:sparse-slackgen} to ensure that all dense vertices have two same-colored neighbors. At a high-level, the parameter $\rho$ is chosen so that we can color nodes with greater sparsity using palette sparsification, and those with less sparsity using information recovered from vectors with $O(\rho)$ non-zero values. By choosing $\rho = \Theta(\eps^{-2}\log n)$, the algorithm uses $\Ot{n\rho} = \Ot{n}$ space while the probabilistic guarantees from slack generation and palette sparsification hold with high probability.

\paragraph{Assumption on $\Delta$.}
We assume $\Delta$ is known in advance and that $\Delta \geq 2\alpha\rho^3$ so that $\exp(-\Omega( \Delta / \rho^2 ))$ and $\exp( - \Omega(\Delta/\rho))$ are at most $1/\poly(n)$.
When $\Delta$ is not known in advance, we run the algorithm $\ceil{\log n}$ times in parallel with $\Delta$ set to each power of two. After the streaming pass, we know $\Delta$ and we thus use the information from the streaming pass corresponding to the largest power of two smaller than $\Delta$.
When $\Delta \leq \rho^3$, we can afford to store all edges of the graph and use that, by \cite{Reed99}, a $(\Delta-1)$-coloring exists for $\Delta \geq \Delta_0$, for some large enough universal constant $\Delta_0$.

\subsection{Palette Sampling}
\label{sec:streaming-palette-sampling}

To avoid storing all edges, our algorithm uses palette sparsification to color some of the vertices. Similarly to \cite{ACK19,AKM23}, we partition the list of colors $\rL(v)$ sampled by $v$ into multiple independent random lists. This is done to simplify the analysis by removing dependencies between the successive coloring steps. The exact sampling process for each of the $\rL_i(v)$ is therefore chosen to be convenient in the analysis of \cref{sec:coloring-alg}, where the index $i$ corresponds to the step at which the coloring algorithm uses $\rL_i(v)$.
For each vertex $v\in V$, let 
\begin{itemize}
    \item $\rL_2(v)$: sample one color in $[\Delta-1]$ uniformly at random,
    \item $\rL_3(v)$: each color of $[\Delta-1]$ is sampled \wp $\rho/\Delta$,
    \item $\rL_4(v)$: each color of $[\Delta-1]$ is sampled independently \wp $\rho/\Delta$, 
    \item $\rL_5(v)$: each color of $[\Delta-1]$ is sampled independently \wp $\rho^2/\Delta$,
    \item $\rL_6(v)$: each color of $[\Delta-1]$ is sampled independently \wp $\rho^3/\Delta$.
\end{itemize}
and $\rL(v)$ is the union of those lists. The random lists of colors are sampled before the beginning of the streaming pass.

During the streaming pass, we store all the edges of the sparsified graph.
\begin{definition}
    \label{def:sparsified-graph}
    For a graph $G=(V,E)$ and random lists $\rL$, the sparsified graph is the subgraph of $G$ with vertex set $V$ and all edges $\set{u,v}\in E$ such that $\rL(u) \cap \rL(v)$ is not empty.
\end{definition}

The following lemma is direct by applications of the Chernoff Bound. We refer readers to \cite[Lemma 4.1]{ACK19} for details. In particular, it shows that \whp the sparsified graph requires only $\Ot{ n }$ space to store.
\begin{lemma}
    \label{lem:sparsified-graph-edges}
    W.h.p., each $\rL(v)$ contains at most $O(\rho^3)$ colors and the sparsified graph contains at most $O(n\rho^6)$ edges.
\end{lemma}

\subsection{Structural Decomposition}
\label{sec:structural-decomposition}
First, partition almost-cliques based on their size.

\begin{definition}
    We say the almost-clique $C$ is \emph{\underline{large}} if $|C| \geq \Delta + 1$, \emph{\underline{small}} if $|C| \leq \Delta + 1 - \rho$ and \emph{\underline{critical}} otherwise.
\end{definition}

Our algorithm focuses almost entirely on coloring critical almost-cliques. The large will be easy to color because they always contain some anti-edges. This motivates the following definition.

\begin{definition}
    \label{def:solitary}
    The \emph{\underline{core of $C$}} is a $K \subseteq C$ such that $G[K]$ is a clique of maximal size. 
    We say $C$ is \emph{\underline{solitary}} iff the core of $C$ has size at most $|C| - 2$.
\end{definition}

If the largest clique of $G[C]$ is not unique, we pick an arbitrary one as the core of $C$. When $C$ is not solitary, the core is uniquely defined unless $G[C]$ contains exactly one anti-edge. We call \underline{core-critical} the vertices in the core of a critical almost-clique. When we have an almost-clique decomposition $\Vsparse, C_1, C_2, \ldots, C_k$, we write $K_i$ for the core of $C_i$ for all $i\in[k]$.
Since $G$ does not contain a $\Delta$-clique, a large almost-clique is solitary. What makes the solitary almost-cliques easy to color is the following fact.

\begin{fact}
    \label{fact:solitary}
    $C$ is solitary iff $G[C]$ contains a 2-anti-matching or a 3-independent set.
\end{fact}

\begin{proof}
    If $G[C]$ contains a 2-anti-matching or a 3-independent set, at least two vertices cannot belong to $K$, hence $C$ is solitary.
    Conversely, if $C$ is solitary, let $u\neq v \in C - K$. Assume that $u$ and $v$ have exactly one anti-neighbor in $K$ and that it is the same vertex $w\in K$, because $G[C]$ otherwise contains a 2-anti-matching. If $u$ and $v$ are adjacent, then $G[K + u + v - w]$ is a clique larger than $K$, which is absurd; hence $\set{u,v,w}$ is a 3-independent set.
\end{proof}

Solitary almost-cliques are such that any coloring of $V - C$ can be extended to $C$ (see \cref{fig:2-anti-matching,fig:3-IS}). The second criteria that determines how we color almost-cliques relates to how they are connected to the vertices outside (see \cref{fig:populaire} for an example of a popular clique and \cref{fig:amicale,fig:reed-transform} for examples of friendly cliques, where $s_1$, $s_2$ and $s$ are friends).

\begin{definition}
    \label{def:popular}
    Let $C$ be a non-solitary almost-clique with core $K$, $v \in N(C) - K$ and $k \in \bR_{>0}$. We say $v$ is a \emph{\underline{$k$-friend}} of $C$ if it has at least $\Delta / k$ neighbors in $K$ and at least one non-neighbor in $K$. 

    If $C$ has two $k$-friend $u$ and $v$ with a shared neighbor in $K$ and there exists $u' \neq v' \in K$ such that $\set{u,u'}$ and $\set{v,v'}$ are not edges, we say that $C$ is \emph{\underline{$k$-popular}}.

    If $C$ has at least one $k$-friend but is not $k$-popular, we say that $C$ is \emph{\underline{$k$-friendly}}.
\end{definition}

\subsection{Sparse Recovery}
\label{sec:streaming-sparse-recovery}
To color the densest regions of the graph, it is necessary to recover some of their structure. For a solitary $C$, the information contained in the solitary helper structure allows us to extend any coloring of $V - C$ to $C$.

\begin{definition}
    \label{def:solitary-helper}
    Let $C$ be a solitary almost-clique. A \emph{\underline{solitary helper structure}} for $C$ is one of the two following types of tuple:
    \begin{itemize}
        \item $(u_1, v_1, u_2, v_2, N_G(u_1), N_G(u_2))$ such that $\set{u_1v_1, u_2v_2}$ is a 2-anti-matching of $G[C]$; or
        \item $(u_1, u_2, v, N_G(u_1), N_G(u_2))$ such that $\set{u_1, u_2, v}$ is a 3-independent set in $G[C]$.
    \end{itemize}
\end{definition}

We do not define other types of helper structure, for we can implicitly recover all the edges incident to the remaining dense vertices. Formally, we recover the following:
\begin{lemma}
    \label{lem:sparse-recovery}
W.h.p., \alg{sparse-recovery} uses $\Ot{ n }$ memory and after the streaming pass, we recover
    \begin{itemize}
        \item the sets $A(v)$ and $E(v)$ for every vertex $v \notin \Vsparse$ with $a(v) + e(v) \leq 4\rho$, and
        \item a solitary helper structure for every non-small solitary almost-clique.
    \end{itemize}
\end{lemma}

We say that ``we know $N_G(v)$'' if we can iterate over the neighbors $u\in N_G(v)$ in time $O(\deg(v))$ and $\poly(\log n)$ space after the streaming pass. So, we know $N_G(u_1)$ and $N_G(u_2)$ where $u_1$ and $u_2$ are vertices from the solitary helper structure described in \cref{def:solitary-helper}. \cref{lem:sparse-recovery} implies that, \whp, we also know $N_G(v)$ for every dense vertex such that $a(v) + e(v) \leq 4\rho$. Indeed, to iterate over the neighbors of $v\in C$, it suffices to iterate over $E(v)$ and then $C - A(v)$, which is possible since $C$, $E(v)$ and $A(v)$ are stored in memory. We remark that we only store $N_G(v)$ implicitly, and hence the space usage remains $\Ot{ n }$.

Since the sparse recovery is almost verbatim that of \cite[Section 4.4]{AKM23}, we skip details here and refer readers to \cref{sec:sparse-recovery-details}.
Observe how using \cref{lem:sparse-recovery} allows us to recover all relevant information for non-small and popular almost-cliques.

\begin{corollary}
    \label{cor:sparse-recovery-popular}
    Let $C$ with core $K$ be critical and not solitary. From \alg{sparse-recovery}, we recover $K$ and $N(v) \cap K$ for all vertices $v \in V$. 
\end{corollary}

\begin{proof}
    Since $C$ is not solitary, we have that $|K| \geq |C|-1$. All vertices of $K$ thus have $a(v) \leq 1$ and $e(v) \leq \Delta + 1- |C| + a(v) \leq \rho + 1 \leq 4\rho$, where the first inequality uses that $C$ is not small. By \cref{lem:sparse-recovery}, we recover all edges with exactly one endpoint in $K$. In particular, if $C \neq K$, we identify the vertex in $C - K$ and thereby recover $K$.
\end{proof}

\subsection{Summary of the Information Collected by the Algorithm}
\label{sec:summary-streaming}

Before moving on with the analysis of our coloring algorithm, we recapitulate all the information gathered during the streaming pass. With high probability, at the end of the stream, we have the following information available:
\begin{itemize}
    \item the random lists of colors $\rL(v)$ as described in \cref{sec:streaming-palette-sampling},
    \item all the edges of the sparsified graph defined in \cref{def:sparsified-graph},
    \item a vertex partition $\Vsparse, C_1, \ldots, C_k$ as in \cref{def:acd} (by \cref{lem:compute-acd-streaming}),
    \item the sets $A(v)$ and $E(v)$ for all $v\in C$ with $a(v) + e(v) \leq 4\rho$ (by \cref{lem:sparse-recovery}),
    \item a solitary-helper structure for every non-small solitary almost-clique (by \cref{lem:sparse-recovery}),
    \item the core of every critical and not solitary almost-clique (by \cref{cor:sparse-recovery-popular}),
    \item the sets of critical $k$-popular and $k$-friendly almost-cliques for all $k > 0$.
\end{itemize}
It is an easy observation that the set of $k$-popular critical almost-cliques can be inferred, such cliques are not solitary and all the edges incident to their core are known.

It should also be clear that the algorithm uses $\Ot{ n }$ space as the sparsified graph contains $O(n\log^6 n)$ edges (by \cref{lem:sparsified-graph-edges}), and sparse recovery uses $\Ot{ n }$ space (by \cref{lem:sparse-recovery}).

\section{The Coloring Algorithm}
\label{sec:coloring-alg}
In this section, we explain how to color the graph after the streaming pass. We have two means to ensure that we do not create monochromatic edges when we color vertices. Either they use a color from their random list $\rL(v)$, and we can check for conflicts by looking at the edges of the sparsified graph (\cref{def:sparsified-graph}). Or we color a vertex with a color that is not from its random list $\rL(v)$, in which case we must be able to iterate over all its neighbors. To emphasize that aspect, we define the notion of serene coloring and maintain that our coloring is one.

\begin{definition}
\label{def:serene}
    We say that a (possibly partial) coloring $\col$ is \emph{\underline{serene}} if for all vertices $v$ with $\col(v) \neq \bot$, either $\col(v) \in \rL(v)$ or we know $N_G(v)$ after the streaming pass.
\end{definition}

\paragraph{Detailed Overview.}
Let us explain the steps of our coloring algorithm. 

\begin{itemize}
    \item \textbf{Step 1: Reed Transform}.
    First, we remove all non-small solitary and non-small  $\alpha\rho^2$-popular almost-cliques from $G$ and obtain the graph $G'$. Then, we remove all $2\rho$-friendly almost-cliques with a $(\Delta-1)$-clique, possibly adding some extra edges to the graph. The resulting graph is called $H$ and has the same almost-clique decomposition (apart from those that got removed), but contains no large, critical solitary, or critical $\rho$-popular almost-clique, and the almost-cliques with a core of size $\Delta-1$ are not $\rho$-friendly.

    \item \textbf{Step 2: Slack Generation.} 
    Every vertex in $H$ gets activated with constant probability and samples a uniform color in $[\Delta-1]$. If that color was sampled by no neighbor, we color the vertex with it.

    \item \textbf{Step 3: Coloring Critical Almost-Cliques of $H$.}
    Recall that all edges incident to the core of critical almost-cliques are known. Hence, we color those vertices serenely without using palette sparsification. We argue that after slack generation, all critical almost-cliques have two vertices whose coloring can be delayed, because (1) they have two pairs of same-colored neighbors, or (2) they have one pair of same-colored neighbor and an uncolored neighbor whose coloring can be deferred even further. Because the coloring of some vertices must be delayed, we actually color critical almost-cliques in two steps.
    
    \item \textbf{Step 4: Coloring Sparse Vertices \& Small Almost-Cliques.}
    After Step 3, all critical almost-cliques have been colored. Observe that in Step 1, we removed all large almost-cliques from $H$, because they are solitary. So it remains to color the small almost-cliques and the sparse vertices. The sparse vertices are colored as in \cite{ACK19}.
    Each dense vertex from this set gets $\Omega(\log n)$ pairs of same-colored neighbors from slack generation. As such, they can be colored using palette sparsification.

    \item \textbf{Step 5: Inverse Reed Transform.} 
    Steps 3 and 4 colored all vertices of $H$. From a proper coloring of $H$, we construct a coloring of $G'$. 
    Step 5 focuses on inverting the Reed Transform, during which we removed $2\rho$-friendly almost-cliques with a $(\Delta-1)$-core from the graph, and possibly added some edge. By same-coloring two pairs of vertices manually, we can extend the coloring to all such almost-cliques. Note that, from then on, we may recolor vertices. In particular, the pairs of same-colored neighbors produced by slack generation are no longer guaranteed to exist.

    \item \textbf{Step 6: Coloring Non-Small Popular \& Non-Small Solitary.}
    We complete the coloring by handling the solitary and friendly almost-cliques removed during the phase 1 of the Reed Transform. Structurally, those almost-cliques have pairs of vertices that, if same-colored, allow us to extend the coloring greedily (recall \cref{fig:degree-1-choosable}). Thanks to sparse recovery, we can identify those pairs and same-color them with colors from their random lists. Then, they can be colored using either a greedy algorithm or palette sparsification.
\end{itemize}

\begin{table}[h]
\centering
\renewcommand{\arraystretch}{1.5}
\setlength{\arrayrulewidth}{1pt} \begin{tabular}{|>{\columncolor{white}}c|c|c|c|c|}
\hline
 & \multicolumn{1}{c|}{\textbf{Solitary}} & \multicolumn{1}{c|}{\textbf{Popular}} & \textbf{Friendly with a $(\Delta-1)$ Core} & \textbf{Others}\\ \hline
\textbf{Small} & \multicolumn{2}{c|}{Step 4} & \cellcolor{gray!50} & Step 4 \\ \hline
\textbf{Large} & Step 6 & \multicolumn{3}{c|}{\cellcolor{gray!50}} \\ \hline
\textbf{Critical} & Step 6 &  Step 6 & Step 5 & Step 3 \\ \hline
\end{tabular}
\caption{For each type of almost-clique, the step of the algorithm at which it is colored. We stress that popular almost-cliques are not solitary and that the ``Friendly with a $(\Delta-1)$ Core'' are neither solitary nor popular.}
\end{table}

\subsection{Step 1: Preprocessing and the Reed Transform}
\label{sec:RT}

This first step removes from the graph all vertices that cannot be colored with slack generation. To ensure that the almost-cliques removed from the graph can be colored later, we may add some edges. The resulting graph is called $H$.

Concretely, we construct the graph $H$ in two phases. First, a preprocessing phase computes the graph $G'$ by removing from $G$ all almost-cliques that are either large, critical solitary, or critical $\alpha\rho^2$-popular. This step is obvious from the information gathered during the streaming pass (\cref{sec:summary-streaming}) and needs not further discussion.
Then, the Reed Transform computes the graph $H$ by removing from $G'$ all $2\rho$-friendly almost-cliques that contain a $(\Delta-1)$-clique and adding selected edges.
We choose the edges to add to $H$ in \cref{alg:RT}. 

\begin{Algorithm}
    \label{alg:RT}
    Reed Transform.

    \paragraph{Input:} 
    a graph $G'$ with no large, critical solitary, or critical $\alpha\rho^2$-popular almost-cliques.

    \paragraph{Output:} a graph $H$ such that: a) almost-cliques with $|K_i| = \Delta-1$ are not $\rho$-friendly and b) solitary and $\alpha\rho^2$-popular almost-cliques are small.

    \medskip

    For each $2\rho$-friendly almost-cliques $C_i$ with a $(\Delta-1)$-core (i.e., $|K_i|=\Delta-1$), do:
    \begin{enumerate}[label=$\roman*)$]
    \item Let $s_i$ be a $(2\rho)$-friend of $C_i$; if $C_i \neq K_i$, take $s_i$ inside $C_i$.
    \item Let $D_i = N_{G'}(s_i) \cap C_i$ be the set of neighbors of $s_i$ in $C_i$.
    \item For each $u \in D_i$, let $f(u)$ denote the external neighbor of $u$ that is not $s_i$, if any.
    \item Sample each vertex into the set $A_{RT}$ independently \wp $p_{RT} = 1/10$. 
    
    \noindent
    Let $S_i = \set{ u \in D_i \cap A_{RT} : f(u) \not\in A_{RT} }$.
    
    \item If all nodes of $D_i$ are of degree $\Delta$ \emph{and} $f(S_i)$ is an independent set then
        \begin{enumerate}
            \item Sample each vertex of $S_i$ into a set $T_i$ \wp $p_{ds} = 1/\rho$. \label{line:sampling-T}
            \item Let $u_i \ne v_i$ be any pair of nodes in $T_i$ such that $x_i = f(u_i)$ and $y_i = f(v_i)$ are distinct, and not members of the same critical almost-clique.
            \label{step:RT-non-critical}
            \item Add the pair $\set{ x_i, y_i }$ to the set $E_{new}$ and add $i$ into the set $I$.
        \end{enumerate}
  \end{enumerate}
  Return the graph $H$ where $V(H) = V(G') - \bigcup_{i \in I} C_i$ and $E(H) = G'[V(H)] \cup E_{new}$.
\end{Algorithm}

\begin{remark}
    Some important observations about \cref{alg:RT} before we analyze it:
    \begin{itemize}
        \item the random set $A_{RT}$ is introduced in \cref{alg:RT} to break-symmetry and simplify the inversion of the Reed Transform; it ensures that an $x$- or $y$-vertex is not selected as a $u$- or $v$-vertex for some other almost-clique and vice versa;
        \item the algorithm can be implemented after the streaming pass because the decomposition is known (\cref{lem:compute-acd-streaming}) and the core as well as all edges incident to core vertices of non-small almost-cliques are known (\cref{cor:sparse-recovery-popular}).
    \end{itemize}
\end{remark}

\begin{lemma}
    \label{lem:RT}
    Assume that $\Delta \geq \alpha\rho^3$ and $\eps \geq \max\set{8/\rho, 4\sqrt{2}/\sqrt{\eta\rho}}$.
    Let $G'$ be a graph with an $(\eta,\eps,\eps)$-almost-clique decomposition $\Vsparse, C_1, C_2,$ $\ldots$, $C_k$ that has no large, critical solitary, or critical $\alpha\rho^2$-popular almost-cliques.
    W.h.p., the graph $H$ produced by \cref{alg:RT} with the same vertex-partition restricted to $V(H)$ is such that:
    \begin{enumerate}[label=(RT.\arabic*),leftmargin=1em,wide]
        \item\label[part]{part:RT-acd} 
        it is an $(\eta/2, \eps,\eps/2)$-almost-clique decomposition;
        \item\label[part]{part:RT-not-sol-pop} 
        all solitary or $\rho$-popular almost-cliques in $H$ are small;
        \item\label[part]{part:RT-not-friendly}
        no almost-clique in $H$ with a $(\Delta-1)$-core is $\rho$-friendly; and
        \item\label[part]{part:RT-nice} $H$ has maximum degree $\Delta$ and does not contain a $\Delta$-clique.
    \end{enumerate}
\end{lemma}

\begin{proof}
We first verify that each $u\in D_i$ of degree $\Delta$ has exactly one external neighbor $f(u)$ different from $s_i$. 
Indeed, $u$ is in the $(\Delta-1)$-clique $K_i$, thus has $\Delta-2$ core neighbors.
Since $u$ is of degree $\Delta$, it has one additional neighbor besides $s_i$, which must be external because $s_i$ is the only vertex of $C_i - K_i$ (when $C_i \neq K_i$).
\begin{fact}
    Every $u\in D_i$ with degree $\Delta$ has exactly one external neighbor $f(u)$ other than $s_i$.
    \label{lem:ext-degs-in-I}
\end{fact}

We need that a constant fraction of $u\in D_i$ belong to $S_i$, i.e., are in $A_{RT}$ while their external neighbor $f(u)$ is not. We show that this is a high probability event.

\begin{claim}
    \label{claim:RT-well-defined}
W.h.p.\ over the randomness of $A_{RT}$, we have $|S_i| \ge |D_i|/24$.
\end{claim}

\begin{proofclaim}
    For $u \in D_i$, let $X_u$ be the indicator random variable of the event that $u \in A_{RT}$ and $f(u)\notin A_{RT}$.
We have $\Exp{X_u} = p_{RT}(1-p_{RT})$ and by the linearity of expectation their sum has expectation $\mu = p_{RT}(1-p_{RT})|D_i| \geq |D_i|/12$. The activation of $u$ affects only $X_u$ and the activation of a vertex in $f(D_i)$ affects at most $\Delta/\alpha\rho^2$ other vertices in $D_i$, for $C_i$ would otherwise be $\alpha\rho^2$-popular. So we may use the read-$k$ bound (\cref{lem:k-read}) with $k = \Delta/(\alpha\rho^2)$ and $\delta=1/24$ to show concentration as
    \[
    \Prob*{ \abs*{ \sum_{u \in D_i} X_u - \mu } \geq |D_i|/24  }
    \leq \exp\paren*{-\Omega\paren*{ \frac{|D_i|}{k} }} 
    \leq \exp(-\Omega(\rho)) \leq 1/\poly(n) \ ,
    \]
    where we use that $|D_i| \geq \Delta/(2\rho)$.
\end{proofclaim}

Henceforth, fix $A_{RT}$ such that \cref{claim:RT-well-defined} holds and let $S = \bigcup_{i \in I} S_i$. 

In \cref{line:sampling-T}, vertices are downsampled to ensure that the algorithm does not add an excessive number of edges to any single vertex.

\begin{claim}
After \cref{line:sampling-T}, over the randomness $T = \bigcup_i T_i$, we have that
\begin{enumerate}[label=$\roman*)$]
\item For each $C_i$, we have that $|T_i| \ge \Delta/(50\rho^2)$ \wp $\exp(-\Omega( \Delta / \rho^2 ))$.

\item For each node $u \in V(G')$, we have that $|f^{-1}(u) \cap T| \le 2 \Delta/\rho$ \wp $\exp(-\Omega(\Delta / \rho))$.
\item For each $C_i$, we have that $|f(T) \cap K_i| \le 4\Delta/\rho$ \wp $\exp(-\Omega(\Delta / \rho))$. 
\end{enumerate}
\label{claim:downsizing}
\end{claim}

\begin{proofclaim}
Consider an almost-clique $C_i$.
Observe that $|T_i|$ is the sum of independent binary variables over the nodes in $S_i$, each with probability $p_{ds}$.
By linearity of expectation, $\Exp{|T_i|} = |S_i| p_{ds} \ge |D_i|/24 \cdot 1/\rho \ge \Delta/(48 \rho^2)$, using \cref{claim:RT-well-defined} and that $|D_i| \geq \Delta/(2\rho)$.
By Chernoff, $|T_i| \ge \Delta/(50\rho^2)$, \wp $\exp(-\Omega(\Delta/\rho^2))$, establishing $i)$.

Consider an arbitrary vertex $u$ of the graph. 
Each neighbor $v \in f^{-1}(u)\cap S$, independently joins $T$ \wp $p_{ds}=1/\rho$. 
Hence, by linearity of expectation, we have $\Exp{|f^{-1}(u) \cap T|} \le \Delta/\rho$.
By Chernoff, we get that $|f^{-1}(u) \cap T|$ is at most $2 \Delta/\rho$ \wp $\exp(-\Omega( \Delta/\rho ))$, establishing $ii)$.

Each vertex of $K_i$ has at most two external neighbors, so $|E(K_i)| \le 2(\Delta+1)$.
Each vertex of $E(K_i)$ joins $T$ independently with probability at most $1/\rho$ so $\Exp{ |E(K_i) \cap T| } \leq 2(\Delta+1)/\rho$. By Chernoff, at most $4\Delta/\rho$ vertices of $E(K_i)$ join $T$ \wp $\exp(-\Omega(\Delta/\rho))$. It implies $iii)$ because each external neighbor in $T$ maps to at most one vertex of $K_i$ through $f$, i.e., $|f(T) \cap K_i| \leq |E(K_i) \cap T|$.
\end{proofclaim}

We deduce that the algorithm is well-defined.

\begin{claim}
W.h.p., for each $C_i$, there is a pair $u_i, v_i$ as required in \cref{step:RT-non-critical} of \cref{alg:RT}.
\label{claim:H-edges}
\end{claim}

\begin{proofclaim}
Let $u_i$ be any node in $T_i$.
If $x_i = f(u_i)$ is in a critical almost-clique $C_j$, then $f(S_i)$ includes at most one other vertex $z \in C_j \cap f(S_i)$ because $f(S_i)$ is an independent set and $C_j$ is not small, hence not solitary. Both $x_i$ and $z$ (if it exists) have at most $\Delta/\alpha\rho^2$ neighbors in $D_i$, for otherwise $C_i$ would be $\alpha\rho^2$-popular.
From \cref{claim:downsizing}-$i)$, we have that $|T_i| \ge \Delta/(50\rho^2)$, w.h.p.
Thus, if we exclude from consideration the nodes in $T_i$ that map by $f$ to either $x_i$ or $z$, we are left
with $|T_i| - 2\Delta/(\alpha\rho^2) \ge \Delta/(\alpha \rho^2) \ge \rho$ suitable candidates (using that $\alpha \geq 150$).
Any of these can be chosen as $v_i$, thereby establishing the claim.
\end{proofclaim}

It is now easy to verify that \ref{part:RT-acd} to \ref{part:RT-not-friendly} of \cref{lem:RT} hold. 
It is immediate from \cref{claim:downsizing}-$ii)$ that at most $2\Delta/\rho$ edges are added into $H$ that are incident on any given vertex.
So, a vertex $v\in \Vsparse$ gains at most $2\Delta/\rho$ new neighbors and each of its previous neighbors gains at most $2\Delta/\rho$ incident edges (by \cref{claim:downsizing}-$ii)$. So the sparsity of $v$ decreases by at most $4\Delta/\rho$, staying above $\eta\eps^2\Delta - 4\Delta/\rho \geq (\eta/2) \eps^2\Delta$. Similarly, we add few enough edges to ensure that \cref{part:acd-ext} of \cref{def:acd} continues to hold, up to a constant factor. 

Using the properties of the almost-clique decomposition \ref{part:RT-acd}, we deduce that $H$ contains no $\Delta$-clique.

\begin{claim}
    The graph $H$ has maximum degree $\Delta$ and contains no $\Delta$-clique.
\end{claim}

\begin{proofclaim}
    The maximum degree remains at most $\Delta$, for each edge added incident to $v$, another edge incident to $v$ is removed. So we focus on proving that $H$ contains no $\Delta$-clique. 

    A $\Delta$-clique in $H$ must be within an almost-clique, since its nodes must be dense, and it can't involve nodes from different almost-cliques (by \cref{def:AC} (\ref{part:acd-neighbors})). But the almost-clique decomposition did not change \ref{part:RT-acd} and no edges were added inside critical almost-cliques (by \cref{step:RT-non-critical}), so the $\Delta$-clique would have to have been contained in $G'$, a contradiction.
\end{proofclaim}

This concludes the proof of \cref{lem:RT}.
\end{proof}

\subsection{Step 2: Slack Generation}
\label{sec:slackgen}
Formally, when we color some graph, the slack of a vertex is the difference between the number of colors it has available and its number of (uncolored) neighbors. For $(\Delta-1)$-coloring, vertices start with $-1$ slack and each time we same-color a pair of neighbors, we increase the slack by one.
We use a well-known technique for generating pairs of same-colored neighbors: let vertices pick random colors. The algorithm is as follows; recall that we run it on the graph $H$ obtained from Step 1.

\begin{Algorithm}
\alg{slack-generation}
\medskip

\textbf{Input:} a graph $H$

\textbf{Output:} a proper partial $(\Delta-1)$-coloring
\medskip

\begin{enumerate}[label=$\roman*)$]
    \item vertices join $A_{SG}$ \wp $p_{SG} = 1/10$.
    \item vertices\footnote{For consistency, set $\chi(v) = \bot$ for every $v \notin A_{SG}$.} of $A_{SG}$ sample $\chi(v)$ uniformly at random in $[\Delta-1]$.
    \item if $\chi(v) \notin \chi(N(v))$, then set $\col(v)$ to $\chi(v)$ and otherwise let $v$ uncolored, i.e., $\col(v) = \bot$.
\end{enumerate}
\label{alg:slackgen}
\end{Algorithm}

To implement \cref{alg:slackgen} after the streaming pass, each vertex uses the one color from $\rL_2$ as the random color $\chi(v)$. It therefore suffices to look at the edges of the sparsified graph to know if a vertex retains its color.

The activation probability ensures that not too many vertices get colored by slack generation.
\begin{fact}
    \label{fact:slackgen-nb-colored}
    After \cref{alg:slackgen}, \whp, every $C_i$ contains at most $\Delta/5$ colored vertices.
\end{fact}
\begin{proof}
    During slack generation, a vertex gets colored only if active. Since every vertex gets activated independently with probability $p = 1/10$, the Chernoff Bound implies that \whp at most $1.5 p|C_i| < \Delta / 5$ are colored in every almost-clique. By union bound, this holds for all almost-cliques of the decomposition.
\end{proof}

We use three flavors of slack generations. First, we have the well-known result (see, e.g., \cite{RM02,EPS15,HKMT21}) that $\zeta$-sparse vertices (cf.~\cref{def:sparse})
receive $\Omega(\zeta)$ slack \wehp in $\zeta$. This explains why sparse vertices and small almost-cliques can be colored easily. 

\begin{proposition}[{\cite[Lemma 6.1]{HKMT21}}]
    \label{thm:slack-gen}
    There exists a universal constant $\gamma \in (0,1)$ for which the following holds.
    Suppose $H$ is a graph of maximal degree $\Delta \gg 1$.
    If $\col$ is the coloring produced by \cref{alg:slackgen}, then a $\zeta$-sparse vertex $v$ with $\zeta \gg 1$ has that
    \[
        |L_\col(v)|
        \geq \deg_\col(v, H) + \gamma \cdot \zeta
    \]
    with probability at least $1 - \exp \paren*{ - \Omega( \zeta) }$.
\end{proposition}

In very dense critical almost-cliques, the sparsity of individual vertices is too small for \cref{thm:slack-gen} to provide vertices with extra colors (by same-coloring neighbors) with high enough probability. However, when the almost-clique has many different external neighbors, a very similar argument shows that many vertices of the almost-clique get one external neighbor colored the same as a vertex inside the almost-clique. For simplicity, we state the result for a clique $K$, rather than an almost-clique.
\begin{proposition}[{Rephrasing Lemma 5.4 in \cite{FHM23}}]
    \label{lem:slackgen-matching}
    Let $K$ be a clique and $M$ a matching between $K$ and $N(K) - K$. After \cref{alg:slackgen}, \wehp in $|M|$, at least $\Omega( |M| )$ vertices in $V(M) \cap K$ have two non-adjacent neighbors colored the same, i.e., one unit of slack.
\end{proposition}

We conclude this anthology of results about slack generation with one from Reed's original paper on $(\Delta-1)$-coloring. It serves the same purpose as \cref{lem:slackgen-matching} except that it counts vertices that receive two units of slack. Consequently, instead of asking for a matching between $C$ and $N(C) - C$, we require a one-to-two matching. Formally, we define a \underline{triple for $C$} as a tuple $(u,v,w)$ such that $v\in C$, $u \neq w \in E(v)$. We say that a set of triples is disjoint if no vertex belongs to more than one.

\begin{proposition}[{\cite[Section 4]{Reed99} or \cite[Chapter 11.3]{RM02}}]
    \label{lem:reed-triples}
    For every $\eps,\delta\in(0,1/10)$, there is a universal constant $\gamma = \gamma(\eps,\delta,p_{SG})$ for which the following holds.
    Let $C$ be an $(\eps,\delta)$-almost-clique and $\set{(u_i, v_i, w_i)}_{i=1}^t$ be a set of disjoint triples for $C$. Let $Z$ be the random variable counting number of triples for which $v_i$ is uncolored, and $u_i$ and $w_i$ are colored the same as some vertex in $C$ by slack generation. Then, 
    \[
    \Prob*{ Z \leq \gamma \cdot t  } \leq 
    \exp\paren*{-\Theta\paren*{ \frac{t^2}{\Delta}} }
    \]
\end{proposition}

Note that the concentration from \cref{lem:reed-triples} is effective only when $t$ is asymptotically larger than $\sqrt{\Delta}$. This is a consequence of the fact that, contrary to \cref{thm:slack-gen,lem:slackgen-matching}, counting colors does not suffice to count the number of vertices with two-units of slack.

\subsection{Step 3: Coloring Critical Cliques in \texorpdfstring{$H$}{H}}
\label{sec:coloring-critical-H}
In this step, we prove that all critical almost-cliques of $H$ can be colored after slack generation.
Since the formal argument is involved, let us provide first a bird's-eye view of the proof. It is in three parts: 
\begin{itemize}
    \item Firstly, in \cref{lem:expanding-critical}, we argue that if a subgraph of $H$ induced by a set of critical almost-cliques is such that each almost-clique has two adjacent vertices with respectively one and two units of slack, we can serenely color that subgraph greedily.
    \item Then we analyze how the critical vertices of $H$ get slack from slack generation (\cref{alg:slackgen}). The analysis is three-fold: when many low-degree vertices exist (\cref{lem:low-degree-critical}), or when many vertices have external neighbors with few edges to the almost-clique (\cref{lem:expanding-critical}), or finally when most of the vertices have degree $\Delta$ and external neighbors with many edges to the almost-clique (\cref{lem:temp-slack}). In the first two cases, we show that \whp at least two vertices get slack. For the last case, we argue that \whp the almost-clique has many vertices with one pair of same-colored neighbors and one uncolored external neighbor with $2\rho$ edges to the almost-clique. 
    \item Finally, we color the critical vertices of $H$ in two stages using \cref{lem:critical-extend}. That is to ensure that, in the last type of critical almost-cliques, the uncolored external neighbors with $2\rho$ neighbors in the almost-clique $C$ remain inactive while we color $C$, thereby providing one more unit of slack. The trick is that such external neighbors, if dense, must have $2\rho$ external neighbors. However, critical almost-cliques of $H$, since they are not solitary, contain at most one such vertex.
\end{itemize}

\noindent
In the end, we get the following result.

\begin{lemma}
    \label{lem:coloring-critical-H}
    Let $H$ be the graph obtained from \hyperref[sec:RT]{Step 1} and $\col$ be the coloring produced by \cref{alg:slackgen} in \hyperref[sec:slackgen]{Step 2}.
    With high probability over the randomness of $\rL_2(V)$ and $\rL_3(V)$, there exists a serene extension of $\col$ to all critical almost-cliques of $H$.
\end{lemma}

\begin{proof}
We shall prove that, after running slack generation in $H$, every critical almost-clique contains two adjacent vertices such that one has two units of slack and the other has one unit of slack. 
The coloring can then be greedily and serenely extended to all vertices of $H$, thereby proving \cref{lem:coloring-critical-H}.

\begin{claim}
    \label{lem:critical-extend}
    Let $J \subseteq \mathbb{N}$ be such that all $C_j$ for $j\in J$ are critical and define $\fC = \bigcup_{j\in J} C_j$.
    Suppose $\psi$ is a partial serene coloring of $H$ such that for all $j\in J$, there are two \emph{adjacent and uncolored} vertices $u \neq v \in K_j$ such that 
    \begin{enumerate}
        \item $|L_\psi(u)| \geq \deg_\psi(u, \fC)$, and
        \item $|L_\psi(v)| \geq \deg_\psi(v, \fC) + 1$.
    \end{enumerate}
    Then, \whp over the randomness of $\rL_3(\fC)$, the coloring $\psi$ can be extended serenely to all vertices of $\fC$.
\end{claim}

\begin{proofclaim}
    We color the almost-cliques of $\fC$ sequentially. Observe that as we extend the coloring to $\fC$, properties (1) and (2) remain true. Consider a critical almost-clique $C_j$. If $C_j \neq K_j$, let $s$ be the unique vertex of $C_j$ not in $K_j$. If $s$ is not colored, we color it as follows. By \cref{fact:slackgen-nb-colored}, it has at least $\Delta/5 - \eps\Delta$ uncolored neighbors, thus at least $\Delta/5 - \eps\Delta -1 \geq \Delta/10$ available colors. The probability that none of the colors from $\rL_3(s)$ is available is at most $(1 - \rho/\Delta)^{\Delta/10} \leq \exp(-\rho/10) \leq 1/\poly(n)$. So we may color $s$ first with a color from $\rL_3(s)$. Then, we greedily color the vertices of $C_j - N(u) \cap N(v) - u -v$, and then those of $C_j - u - v$. Since $u$ is adjacent to $v$, which is uncolored, we may color $u$ next. Finally, the vertex $v$ always has at least one available color, so it may be colored last.
    Observe that, by \cref{lem:sparse-recovery}, we know $N_G(v)$ and $N_H(v)$ for all $v\in K_i$. In particular, our coloring remains serene.
\end{proofclaim}

For the second part of the proof, we analyze how critical almost-cliques of $H$ get slack from slack generation. Consider a fixed critical $C = C_i$ of $H$ with core $K = K_i$. Let us begin with the case where $K$ contains many vertices of degree less than $\Delta$.

\begin{claim}[Low-Degree Vertices]
    \label{lem:low-degree-critical}
Suppose $S \subseteq K$ is a set of at least $\Delta/10$ vertices with degree less than $\Delta$ in $H$. Then after slack generation (\cref{alg:slackgen}), \whp, there are two vertices $u \neq v \in S$ as in \cref{lem:critical-extend}.
\end{claim}

\begin{proofclaim}
If at least one vertex of $S$ has degree at most $\Delta - 2$ and another has degree at most $\Delta-1$, the claim already holds. So we may assume without loss of generality that all vertices of $S$ have degree $\Delta - 1$. 
We show that there exists a matching of size at least $\rho/10$ between vertices of $S$ and $C - N(C)$. \cref{lem:slackgen-matching} then implies that, \whp, at least two vertices $u\in S$ have an external neighbor colored the same as a neighbor in $K$ by slack generation. As such, at least two vertices in $K$ have $|L_\col(v)| \geq \Delta - 1 - |\col(N_H(v))| \geq \Delta - 1 - |N_H(v) \cap \dom\col| + 1 \geq \deg_\col(v,H) + 1$.

To prove the existence of the matching, we argue that each vertex of degree $\Delta-1$ has at least one external neighbor with fewer than $\Delta/\rho$ neighbors in $K$. A greedy construction therefore produces a matching of $\frac{\Delta/10}{\Delta/\rho} \geq \rho/10$ edges $\set{u, v}$ where $u\in S$ and $v\in N(C) - C$.

There are three cases.
\begin{itemize}
\item If $|C| = |K| = \Delta-1$. The vertices of degree $\Delta-1$ have external degree exactly $\Delta-1 - (|K| - 1) = \Delta - |K| = 1$. Since $C$ has no $\rho$-friend \ref{part:RT-not-friendly}, the external neighbor has fewer than $\Delta/\rho$ edges to $K$.

\item If $|C| = |K| \leq \Delta-2$, then a vertex of degree $\Delta-1$ has exactly $\Delta-1 - (|K| - 1) \geq 2$ external neighbors. Since $C$ is not $\rho$-popular \ref{part:RT-not-sol-pop}, at most one of these external neighbors may have $\Delta/\rho$ edges to $K$.

\item If $|C| - 1 = |K| \leq \Delta-2$, call $s$ the vertex of $C - K$. A vertex of degree $\Delta-1$ in $K - N_H(s)$ has at least two external neighbors, of which at most one can have $\Delta/\rho$ edges to $K$. A vertex of degree $\Delta-1$ in $K \cap N_H(s)$ has at least one external neighbor, which cannot have $\Delta/\rho$ edges to $K$ for otherwise $C$ would be $\rho$-popular, with $s$ and that external vertex.
\end{itemize}
These are all possible cases since $H$ contains no $\Delta$-clique and $C$ is not solitary.
\end{proofclaim}

The remaining two cases depend on how $C$ connects to the outside.
Partition external neighbors of $C$ depending on their number of neighbors inside $K$:
\[
    X_C = \set{ u\in N(C) - C : |N(u) \cap K| \geq 2\rho }
    \quad\text{and}\quad
    Y_C = N(C) - (C \cup X) \ .
\]
First, observe that cliques where most vertices have two external neighbors in $Y_C$ are easy.

\begin{claim}[Expanding Cliques]
    \label{lem:expanding-critical}
    If there is a set $S \subseteq K$ of at least $\Delta/2$ vertices with two external neighbors in $Y_C$, then after slack generation, $\Omega(\Delta/\rho)$ vertices of $S$ have $|L_\col(v)| \geq \deg_\col(v, H) + 1$, with high probability.
\end{claim}

\begin{proofclaim}
    Construct greedily disjoint triples $(u, v, w)$ where $v\in K$ and $u\neq w\in E(v) \cap Y_C$. Since $u$ and $w$ have at most $2\rho$ neighbors in $K$, we can find at least $t = (\Delta/2)/(4\rho) = \Delta/(8\rho)$ such disjoint triples. By \cref{lem:reed-triples}, after slack generation, at least $\Omega(t)$ triples are such that their $v$-vertex has two pairs of same-colored neighbors \wehp in $t^2 / \Delta = \Omega( \Delta / \rho^2 )$; hence \whp for $\Delta \geq \rho^3$. Simple accounting shows that such vertices $v$ verify $|L_\col(v)| \geq \deg_\col(v, H) + 1$.
\end{proofclaim}

For the last case, we do not show directly that some vertex gets two pairs of same-colored neighbors. Instead, we argue that there must be many vertices with one pair of same-colored neighbor and one \emph{uncolored} external neighbor in $X_C$ whose coloring can be deferred.

\begin{claim}
    \label{lem:temp-slack}
    Suppose there is a set $S \subseteq K$ of at least $\Delta/3$ vertices of degree $\Delta$ with at most one external neighbor in $Y_C$. 
    Then, \whp, at least $\Omega(\rho)$ vertices in $S$ have an uncolored neighbor in $X_C$ and an external neighbor colored the same as a vertex in $K$ by slack generation.
\end{claim}

\begin{proofclaim}
    Number the vertices of $S$ as $u_1, u_2, \ldots, u_\ell$ where $\ell = |S| \geq \Delta/3$.
    First, we show that every $u_i$ has neighbors $x_i \in E(u_i) \cap X_C$ and $z_i \in E(u_i) - x_i$ with $|N_H(x_i) \cap K|, |N_H(z_i) \cap K| \leq \Delta/\rho$.
    The proof is by case analysis. 
    
    If $|K| = \Delta-1$, then $C = K$ and $C$ has no $\rho$-friend \ref{part:RT-not-friendly}, hence every $u_i$ has two external neighbors, at least one of which belongs to $X_C$ (since $u_i$ would otherwise have more than one external neighbor in $Y_C$). Suppose now that $|K| \leq \Delta - 2$. If $K = C$, then every $u_i$ has $\Delta - (|K|-1) \geq 3$ external neighbors. At least two belong to $X_C$ (since $u_i$ would otherwise have more than one external neighbor in $Y_C$) and at most one can have more than $\Delta/\rho$ edges to $K$ (since $C$ is not $\rho$-popular, by \ref{part:RT-not-sol-pop}). If $K = C - s$, the case where $u_i \in K - N(s)$, is similar. When $u_i \in K \cap N(s)$, then $u_i$ has $\Delta - (|C|-1) \geq 2$ external neighbors. Neither of them has $\Delta/\rho$ edges to $K$, for $K$ would otherwise be $\rho$-popular (because of $s$) and at most one of them belong to $Y_C$. So we find suitable $x_i$ and $z_i$ for all $u_i$.

    Now, we show that after slack generation, \whp, at least $\Omega( \Delta )$ indices $i \in[\ell]$ are such that $u_i$ and $x_i$ are uncolored, and $z_i$ is colored the same as some vertex in $K$. 
    First, we argue that, after Step (1) of slack generation, \whp over the randomness of $A_{SG}$, at least $\Omega( \ell )$ vertices $u_i$ are not in $A_{SG}$, have $x_i \notin A_{SG}$ and $z_i\in A_{SG}$.
    Let $X_i$ be the indicator random variable of the event that
    $u_i, x_i \notin A_{SG}$ and $z_i \in A_{SG}$. 
    Clearly, $\Exp{ X_i } = (1-p_{SG})^2 p_{SG}$. The activation of some $u_i$ affects only $X_i$; activation of other vertices in $C$ are irrelevant. The activation of a vertex $x_i,z_i$ outside $C$ affects all neighboring vertices in $C$, which is at most $k = \Delta / \rho$ by our choice of $x_i$ and $z_i$. Let $X = \sum_i X_i$, by the read-$k$ bound (\cref{lem:k-read}) with $\delta = (1-p_{SG})^2 p_{SG} /2$, we have that
    \[
    \Prob*{ \abs*{ X - \Exp{X} } > \Exp{ X }/2 }
    \leq \exp\paren*{ -\frac{\paren*{ (1-p_{SG})^2 p_{SG}/2 }^2 \cdot 2\ell}{k} }
    \leq 1/\poly(n) \ .
    \]
    Thus, \whp, the set $U = \set{u_i : u_i, x_i \notin A_{SG}, z_i\in A_{SG}}$ of such vertices has size at least $(1-p_{SG})^2 p_{SG} \ell/2 = \Omega( \Delta )$.

    Call $R$ the set of activated external neighbors $z_i \in A_{SG}$ and consider the bipartite graph with bipartition $(U, R)$ and an edge between $u\in U$ and $v\in R$ iff the edge exists in $H$. 
    As we argued, $U$ has size $\Omega(\Delta)$ and, by definition, each $u_i \in U$ has at least one edge to $z_i \in R$. Conversely, every vertex in $R$ has at most $\Delta / \rho$ neighbors in $U \subseteq K$, so there is a matching of size at least $\Omega(\rho)$ in this bipartite graph. Using the \cref{lem:slackgen-matching}\footnote{In \cite{FHM23}, the proof of \cref{lem:slackgen-matching} requires that vertices be activated randomly with probability, say, $1/2$. Given that we condition on the activation of slack generation $A_{SG}$ earlier in the proof, it might not be clear that the argument still applies. {A simple fix consist in having two shots of activation: view the $1/10$ activation probability in \cref{alg:slackgen} as the result of the random experiment where we activate first with probability $1/5$, and then again with probability $1/2$.} The argument from this proof conditions on the first activation bit, while \cref{lem:slackgen-matching} uses the randomness of the second activation.}, we get that at least $\Omega(\rho)$ vertices $u\in U$ adjacent to $z\in R$ such that $z$ is colored the same as a vertex in $K$.
\end{proofclaim}

Now that we have established how each type of critical almost-clique gets slack, let us describe how we serenely extend the coloring produced by slack generation.
We color critical almost-cliques of $H$ in two stages.

\underline{Stage 1:}
Let $J_1$ be the set of indices such that $C_j = K_j$ or the unique vertex in $C_j - K_j$ has external degree strictly less than $2\rho$ or was colored by slack generation. For the $j\in J_1$ such that $C_j$ is as described by \cref{lem:low-degree-critical,lem:expanding-critical}, \whp, the vertices described in \cref{lem:critical-extend} (with $J = J_1$ and $\fC_1 = \bigcup_{j\in J} C_j$) exist. 
Otherwise by \cref{lem:temp-slack}, \whp, all $C_j$ with $j\in J_1$ contain at least $\Omega(\rho) > 2$ vertices $u \in K_j$ with a neighbor $x \in X_{C_j}$ uncolored in $\col$ and a neighbor $z$ colored the same as a vertex of $K_j$.
Crucially, we have that $x \notin \fC_1$ because the core vertices in critical almost-cliques have external degree at most $\rho+1 < 2\rho$ --- which is less than necessary to belong to $X_{C_j}$ --- and the only non-core vertices with external degree $\geq 2\rho$ included in $\fC_1$ are colored by slack generation --- while $x$ is uncolored.
Consequently, the number of colors available to $u$ is 
\begin{align*}
    |L_\col(u)| &= \Delta - 1 - |\col(N_H(u))|\\
    &\geq \Delta - |N_H(u) \cap \dom\col| \tag{because $u$ has two neighbors colored the same} \\
    &\geq \deg_\col(u, H) \geq \deg_\col(u, \fC_1)+1 \ .
    \tag{because $\deg(u,H)=\Delta$ and $x \notin \fC_1$}
\end{align*}
All almost-cliques of $\fC_1$ verify the conditions of \cref{lem:critical-extend}, hence we serenely extend the coloring produced by slack generation such that all vertices of $\fC_1$ are colored.

\underline{Stage 2:}
Let $J_2$ be the remaining indices $j$ such that $C_j$ is critical. In particular, for $j\in J_2$, we have that $C_j - K_j$ contains a vertex $s_j$ uncolored after Stage 1. Call $D = K_j \cap N(s_j)$. If $D$ verifies the assumptions of either \cref{lem:low-degree-critical} or \cref{lem:expanding-critical}, \whp, at least two vertices verify the assumptions of \cref{lem:critical-extend}. 
If this is not the case, we can construct a matching of $\frac{(1-\eps)\Delta - \Delta/10 - \Delta/2}{\Delta/\rho} \geq \rho/3$ vertices between $D$ and $N(C) - C$ because each vertex outside $C$ has at most $\Delta/\rho$ neighbors in $D$ (otherwise $C$ would be $\rho$-popular). By \cref{lem:slackgen-matching}, \whp, at least $\Omega(\rho)$ vertices of $D$ have a pair of same-colored neighbors. 
We claim that, \whp, one anti-neighbor of $s_j$ was not colored by slack generation. Indeed, since $e(s_j) \geq 2\rho$ (because $j$ would otherwise be in $J_1$), $s_j$ has at least $|C_j| - (\Delta+1) + e(s_j) \geq \rho$ anti-neighbors, where the last inequality uses that $C_j$ is not small. The probability that they are all activated during slack generation is at most $p_{SG}^{\rho} \leq 1/\poly(n)$, and hence \whp at least one $z_i \in K_j - N(s_j)$ is uncolored after Stage 1. We same-color $s_j$ and $z_j$ using a color from $\rL_3(s_j)$, which exists \wp at least $1 - (1 - \rho/\Delta)^{\Delta/3} \geq 1 - 1/\poly(n)$. Then every vertex of $D$ has two pairs of same-colored vertices in their neighborhood. Hence, the vertices of $\fC_2 = \bigcup_{j\in J_2} C_j$ can be colored by \cref{lem:critical-extend}. 

This concludes the proof of \cref{lem:coloring-critical-H} for $\fC_1$ and $\fC_2$ consists of all vertices in critical almost-cliques of $H$.
\end{proof}

\subsection{Step 4: Coloring Sparse Vertices \& Small Almost-Cliques}
\label{sec:coloring-H-not-critical}
We color the non-critical vertices of $H$ in the following order: first, the sparse vertices; second, the small almost-cliques with fewer than $10^7\eps\Delta$ anti-edges; last, the small almost-cliques with at least $10^7\eps\Delta$ anti-edges.

\paragraph{Coloring Sparse Vertices.}
By \cref{thm:slack-gen}, has $\deg_\col(v, H) + \Omega(\eps^2 \Delta)$ available colors after slack generation. Again, as the coloring has only been extended, this remains true in Step 4 after we have colored all almost-cliques of $H$. Then, the greedy algorithm colors the sparse vertices with high probability. Indeed, when it comes to $v$, whatever the current coloring may be, the probability that none of $\Omega(\eps^2\Delta)$ colors available to $v$ get sampled in $\rL_4$ is at most $\paren*{ 1 - \frac{\rho}{\Delta} }^{\Omega(\eps^2\Delta)} \leq \exp(- \Omega( \eps^2\rho) ) \leq 1/\poly(n)$, from our choice of $\rho$. It extends to all sparse vertices by union bound.

\paragraph{Coloring the Small Almost-Cliques.}
We use that every dense vertex in a small almost-clique $C$ is $\Omega(\eps^2 \rho)$-sparse. Indeed, if $e(v) \geq \rho/2$, it follows from \cref{fact:acd-dense-sparsity}, and otherwise $v$ has degree at most $|C| - 1 + e(v) - a(v) \leq \Delta - \rho/2$, hence is $\Omega(\rho)$-sparse. By \cref{thm:slack-gen} about slack generation, that after \cref{alg:slackgen}, \whp every $v$ in a small almost-clique has
\begin{equation}
    \label{eq:sparse-slackgen}
    |L_\col(v)| 
    \geq \deg_\col(v,H) + \Omega( \eps^2\rho ) 
    \geq \deg_\col(v,H) + \Omega(\log n) \ ,
\end{equation}
where the second inequality uses that $\rho = \Theta(\eps^{-2}\log n)$.
Note that the coloring produced by \hyperref[sec:coloring-critical-H]{Step 3} is an extension of the coloring produced by slack generation, thus that \cref{eq:sparse-slackgen} continues to holds at the beginning of Step 4. We now argue that the coloring can be extended with palette sparsification.

Recall that almost-cliques with $10^7\eps\Delta$ anti-edges can be colored using palette sparsification, by the argument of \cite{ACK19} stated here as \cref{lem:streaming-color-trouée}. We stress that the small almost-cliques with many anti-edges are colored last, because \cref{lem:streaming-color-trouée} assumes that vertices get uncolored, which means we might have to recolor some vertices colored by slack generation. In any case, it suffices to prove the following.

\begin{lemma}
    \label{lem:streaming-color-small}
    Fix a small almost-clique $C$ with fewer than $10^7\eps\Delta$ anti-edges. With high probability over $\rL_2(V)$ and $\rL_4(C)$, any extension of the coloring computed by slack generation can be further extended to $C$ using only colors from $\rL_4(v)$ for $v\in C$.
\end{lemma}

\begin{proof}
Consider such a $C$.
Let $\fL$ be the uncolored vertices of $C$ and $\fR$ the colors used by none of the vertices in $C$. Let $\fG$ be the palette graph: the bipartite graph with bi-partition $\fL \cup \fR$ and an edge between $v\in \fL$ and $\chi\in \fR$ when $\chi$ is available to $v$, with respect to the coloring we computed so far. We use \cref{lem:palette-graph} to show that $\fG$ where each edge is retained with probability $\rho/\Delta$ has a $\fL$-perfect-matching, thus the coloring can be extended to every vertex of $C$, with high probability.
Call $t$ the number of vertices in $C$ colored by slack generation.

Call $k = |\fL| = |C| - t$.
Every color available to $v$ belongs to $\fR$ unless it is used by an anti-neighbor. Hence, every $v\in \fL$ has $\deg(v, \fR)
    \geq |L_\col(v)| - a(v)
    \geq k + \Omega(\log n) - 2 a(v)$,
where the second inequality uses \cref{eq:sparse-slackgen} along with the fact that $v$ has at least $k - a(v)$ uncolored neighbors in $C$.

Now, let us explain how we apply \cref{lem:palette-graph}. Since $C$ contains fewer than $10^7\eps\Delta \leq \Delta/10$ anti-edges, from picking $\eps \leq 10^{-8}$, at least half of the vertices have no anti-neighbor. For such a vertex $v$, the lower bound on $\deg(v,\fR)$ implies that $|\fR| \geq \deg(v, \fR) \geq k$.
The inequality $|\fR| \leq 2k$ holds because $|\fR| \leq \Delta - 1$ and $k = |C| - t \geq 4\Delta/5 \geq |\fR|/2$ by \cref{fact:slackgen-nb-colored}. This shows the Part (1) of \cref{lem:palette-graph}. Part (2) follows from the lower bound on $\deg(v,\fR)$ since $k \geq \Delta/2$ by \cref{fact:slackgen-nb-colored}, and $a(v) \leq \eps\Delta \leq 2\eps k \leq k/3$ for every vertex.
We conclude by proving Part (3). Fix $S \subseteq A$ such that $|S| \geq k/2$. Then, using again the lower bound on $\deg(v,\fR)$, we have
\begin{align*}
    \sum_{v\in S} \deg(v,\fR)
    &\geq |S| \cdot k + |S| \cdot \Omega( \log n )  - 2\sum_{v\in S} a(v) \\
    &\geq |S| \cdot k + |S| \cdot \Omega( \log n ) - 2 \cdot 10^7\eps\Delta
    \geq |S| \cdot k \ ,
\end{align*}
where the second inequality uses the bound on the number of anti-edges in $C$ and the last that $\eps \leq 10^{-8}$.
\end{proof}

\subsection{Step 5: Inverting the Reed Transform}
\label{sec:IRT}
In this step, we take the coloring of $H$ computed thus far and inverse the Reed Transform from \cref{alg:RT} to deduce a coloring of $G'$, i.e., of the all vertices except those in non-small solitary or non-small $\alpha\rho^2$-popular almost-cliques. We emphasize that we may change colors given by slack generation (\cref{alg:slackgen}), thus guarantees of \cref{thm:slack-gen,lem:slackgen-matching,lem:reed-triples} do not apply henceforth.

The algorithm same-color two pairs of vertices in each almost-clique removed by the Reed Transform. This guarantees that the remaining vertices can be serenely colored (recall that we know the neighborhood of core-critical vertices). We argue that the pairs can be same-colored greedily, in fact they always have $\Omega( \Delta / \rho )$ colors available. When certain pairs of vertices to same-color overlap, then we same-color entire independent sets at once.

\begin{lemma}
    \label{lem:IRT}
    Let $H$ be the graph from \cref{alg:RT}. Suppose $\col$ is a total proper serene coloring of $H$. With high probability, we compute a serene coloring of $G'$, i.e., of the graph $G$ where all non-small solitary and non-small $\alpha\rho^2$-popular almost-cliques are removed.
\end{lemma}

\begin{proof}
The goal is to same-color all pairs of vertices $x_i v_i$ and $s_i z_i$, where $z_i$ is an anti-neighbor of $s_i$ in $K_i$. Then the coloring can be serenely extended to all $2\rho$-friendly critical almost-cliques with a $(\Delta-1)$-clique, by \cref{lem:critical-extend}. If a vertex in $D_i$ does not have degree $\Delta$, we do not have any vertices $x_i$ and $y_i$; same-coloring the $s_iz_i$ pair suffices. If $f(S_i)$ is not an independent set, we pick any pair $u_i \neq v_i \in S_i$ such that $x_i = f(u_i)$ and $y_i = f(v_i)$ are adjacent in $G'$.

Denote by $\cS$, $\cX$, $\cY$, $\cU$ and $\cV$ the set vertices respectively containing all $s_i$, $x_i$, $y_i$, $u_i$ and $v_i$. Recall that \cref{alg:RT} chooses $\cX$ and $\cV$ disjoint (because $\cV \subseteq A_{RT}$ while $\cX \cap A_{RT} = \emptyset$).
For $w$, call $S(w)$ the set of $i\in I$ for which $w = s_i$. First, let us explain how we select the $z_i$ vertices.

\begin{restatable}{claim}{ClaimZi}
    \label{lem:nice-zi}
If $|S(w)| \geq 1$, then there exist vertices $z_i \in K_i - N(w)$ for all $i\in S(w)$ and such that 
    \[
    \fI(w) := 
    \set{ w } \cup \set{ z_i : w = s_i } \cup \set{ v_i : w = x_i \text{ or } z_j = x_i \text{ for some $j\in S(w)$} }
    \]
    is an independent set.
\end{restatable}

\begin{proofclaim}
There are two cases.

\underline{Case 1:} $w\in C_i$ and $S(w) = \set{i}$.
    Pick any $z_i\in K_i - N(w)$. Let us argue that $\fI(w)$ is an independent set. First, $z_i$ is not adjacent to a $v_j$ with $w = x_j$ because $v_j$ has only one external neighbor different from $s_j$, and it is $w$. Consider a $v_j$ with $z_i = x_j$. There are no edges between $v_j$ and $w$ because it would mean that $w = y_j$ hence that the Reed Transform chose $x_j, y_j$ in the same critical almost-clique, which is impossible (by \cref{alg:RT}, \cref{step:RT-non-critical}). Also note that $\cV$ is an independent set (because $\cV \subseteq A_{RT}$ while $f(\cV) \cap A_{RT} = \emptyset$).

\underline{Case 2:} Observe that, if we are not in Case 1, we have that $K_i - N(w)$ contains at least $\rho$ vertices for each $i\in S(w)$. It follows from \cref{part:acd-ext} of \cref{def:AC} when $w \notin C_i$ and, when $w \in C_i$, from the observation that $w$ has at least $e(w) \geq \Delta/(2\rho)$ external neighbors (because $|S(w)| \geq 2$ hence $w$ is $\Delta/(2\rho)$-friendly with an almost-clique different from $C_i$), thus $w$ has at least $a(v) \geq e(v) - \rho \geq \rho$ anti-neighbors (the first inequality uses that $C_i$ is not small).

    We use the probabilistic method.
    Sample each node of $\bigcup_{i \in S(w)} K_i - N(w)$ into a set $R$ \wp $1/10$.
    For each $i \in S(w)$, let $R_i = R \cap K_i - N(w)$.
    Define $X_u$ as the indicator random variable equal to one iff $u \in R$ and $E(u) \cap R = \emptyset$ and $f(v_j) \notin R$ for all $v_j$ with $u = x_j$. Observe that if we only select $z_i$'s such that $X_{z_i} = 1$, then we are guaranteed that $\fI(w)$ is an independent set. Since $u$ has at most two external neighbors, each leading to at most one $f(v_j)$, it should be clear that $\Exp{ X_u } \ge 1/10(1-4/10) \geq 1/100$. The variable depends on vertices of external degree at most two joining $R$ or not. As each vertex in some $K_i - N(w)$ has at most two external neighbors, the read-$k$ bound (\cref{lem:k-read} with $k=O(1)$ and $\delta=1/200$) implies that $\sum_{u\in K_i - N(w)} X_u \geq \rho / 200$ \wehp in $\rho/k = \Omega( \rho )$, hence with high probability in $n$.
    
    Taking the union bound over the at most $\Delta$ indices in $S(w)$, we can pick all the $z_i$ this way, \wp $1-\exp(-\Omega(\rho))$. In particular, there must exist such a choice of $\{z_i\}$, as desired.
\end{proofclaim}

Call $z_i$ the vertex given by \cref{lem:nice-zi} in $C_i$ and  call $\cZ$ the set of all $z_i$ vertices. 
Note that $\cX$ can overlap with $\cS$ and $\cZ$.
The nodes in $\cZ,\cV$ are core-critical but not those in $\cS$. 
We same-color the pairs in the following order:
\begin{enumerate}
    \item uncolor all vertices of $\cS$,
    \item color all $v_i$ with $\col(x_i)$ where $x_i$ is colored,
    \item color all pairs $s_i z_i$ and $x_i v_i$ where $x_i \in \cS \cup \cZ$, and finally
    \item color all remaining pairs $x_i v_i$, i.e., where $x_i \notin \cS \cup \cZ$ and $x_i\in C_j$ with $j\in I$.
\end{enumerate}
Step (2) cannot create a color conflict because $v_i$ has at most one colored external neighbor $y_i$ (after uncoloring $S$) and $H$ (from the Reed Transform) contained an edge between $x_i$ and $y_i$.

Let us explain how we perform Step (3). We form a virtual graph $Q_1$ by contracting sets of vertices in $G$. More precisely, for a set of \emph{disjoint} sets of vertices $\mathscr{I} \subseteq 2^{V(G)}$, denote by $G/\mathscr{I}$ the graph with vertex set $\mathscr{I}$ and an edge between two vertices if the corresponding sets contain adjacent vertices in $G$. Let $\mathscr{I}_1$ be the sets of $\fI(w)$ for all $w\in S$ and $Q_1 = G/\mathscr{I}_1$. For a vertex $\fI$ of $Q_1$ (i.e., $\fI \in \mathscr{I}_1$), let $L(\fI)$ be the intersection of the lists of available colors in $G$ overall the vertices in $\fI$, i.e., $L(\fI) = \bigcap_{w\in \fI} L_\col(w)$. 

Since each $\fI(w)$ is an independent set, a coloring of $Q_1$ using lists $L(I)$ induces a proper extension of the coloring to $G$ where vertices from the same $\fI(w) \in \mathscr{I}_1$ are colored the same.
Since each $s_i$ has at least $\Delta/\rho$ uncolored neighbors, while each element $\fI(s_i)$ other than $s_i$ gives rises to $O(1)$ edges, we do obtain a $(\deg+ \Omega(\Delta/\rho))$-list-coloring instance.

\begin{claim}
    \label{claim:s-d1LC}
    For all $\fI \in \mathscr{I}_1$, we have that $|L(\fI)| \geq \deg(\fI, Q_1) + \Delta / 4\rho$.
\end{claim}
\begin{proofclaim}
    By definition, $\fI$ contains exactly one vertex $s \in S$ and the remaining vertices are from $\cZ$ or $\cV$.
    Let $W$ be the number of neighbors of $s$ (in $G$)
    that have already been colored. Let $T$ be the number of uncolored neighbors of $s$ that are not in an almost-clique containing a vertex of $\fI$.
    Let $R = |S(s)|$ be the number of pairs $s_j z_j$ such that $s = s_j$.
    For each of the $R$ pairs $s z_i$ (for which $s = s_i$), the vertex $s$ has $\Delta / 2\rho$ neighbors in $K_i$ (as a $2\rho$-friend of $C_i$ in $G$).
    Let $P$ be the number of pairs $s v_i$ such that $s = x_i$.
    For each of the $P$ pairs $s v_i$ (for which $s = x_i$), $s$ has a neighbor $u_i$ in $C_i$. 
    Since the total degree of $s$ in $G$ is at most $\Delta$, we have that 
    \[ T + R \cdot \Delta/(2\rho) + P + W \le \Delta\ . \]

    To count the degree of $\fI$ in $Q_1$, we count the number of edges between vertices of $I$ and vertices in some $\fI' \in \mathscr{I}_1 - \fI$. First, each of the $T$ uncolored neighbor of $s$ could belong to some $\fI'$. Second, each $v_i$ with $s = x_i$ has one edge to $s_i \notin \fI$. Third, for each $j\in S(s)$, the vertex $s$ is adjacent to $v_j \notin \fI$ that gets same-colored as some $x_j$. Finally, each $z_i \in \fI$ has at most two external neighbors that each cause to increase of $\fI$ by one, either because the external is an uncolored vertex in some other $\fI'\in \mathscr{I}_1$ or because it is a $v_j$ such that $z_i = x_j$, and it has an edge to $s_j$. So we may bound the degree in $Q_1$ as 
    \[
        \deg(\fI, Q_1) \leq T + 3R + P \ .
    \]

    The $L(\fI)$ contains at least $\Delta - 1 - W - 2R$ colors because a pair $x_iv_i$ where $x_i = w$ or $x_i = z_j$ is such that all neighbors of $v_i$ are uncolored.
    Combining the two displayed expressions, we obtain that $|L(\fI)| \geq T + P + R (\Delta/(2\rho) - 3) - 1$. Since $R \ge 1$ and $\Delta \geq \rho^3$ by assumption, the claim follows.
\end{proofclaim}

To implement this in streaming, recall that we know all edges incident to core-critical vertices, which include $\cZ$ and $\cV$. In particular, we can compute the sets $\fI(w)$ described in \cref{lem:nice-zi}, which thereby implies that we know $Q_1$.
We color vertices of $Q_1$ greedily.
When it comes to $\fI = \fI(w)$ with $w\in S$, we claim that we can color $\fI$ with some color from $\rL_5(w)$ with high probability. Indeed, by \cref{claim:s-d1LC}, for any partial coloring of $Q_1$, the vertex $\fI$ still has $\Omega(\Delta/\rho)$ available colors, thus
\[
\Prob*{ \rL_5(w) \cap L(\fI) = \emptyset } 
\leq \paren*{ 1 - \frac{\rho^2}{\Delta} }^{\frac{ \Delta }{ 4\rho }}
\leq \exp(-\Omega(\rho)) 
\leq 1/\poly(n) \ .
\]
In particular, when we color vertices of $\fI$ with that color, the coloring remains serene because (1) $w$ is the only non-core-critical vertex of this set, and it gets a color from its random list, and (2) we know the neighborhood of all other vertices from sparse recovery.

The pairs that remain to be same-colored are of the form $x_i v_i$ such that $x_i \notin \cS \cup \cZ$ and $x_i \in C_j$ for some $j\in I$. As such, it has at least $\Delta-O(\rho)$ available colors and at most two external neighbors. It can be that $x_i = x_j = x$, in which case we must same color the three vertices $x$, $v_i$ and $v_j$. Note that $v_i$ and $v_j$ cannot be adjacent as $\cV$ is an independent set. Since the uncolored vertices in $\cX$ and $\cV$ belong to $(\Delta-1)$-cliques of which at most two vertices are colored, the vertices to same-color share at least $\Delta - O(1)$ available colors. 

\cref{claim:downsizing}-$iii)$ implies that $|\cX \cap C_i| \leq 4\Delta/\rho$ for all $i\in I$ because $\cX \subseteq f(T)$.
Hence, the degree of virtual vertices is at most $12\Delta/\rho$, and we can use $(\deg+1)$-list-coloring. Recall that all the $x_i$ that remains are core-critical, hence we know their neighborhood and can serenely color them with any color.
\end{proof}

\subsection{Step 6: Post-processing}
Finally, we assume that $\col$ is a serene coloring of $G'$, i.e., of all the vertices except for the non-small solitary and non-small $\alpha\rho^2$-popular almost-cliques. We iterate over such almost-cliques and apply either \cref{lem:color-popular} or \cref{lem:color-solitary}. The algorithm successfully colors all such almost-cliques with high probability by the union bound.

\begin{lemma}
    \label{lem:color-popular}
    Suppose $\col$ is a serene coloring and $C$ is an $\alpha\rho^2$-popular almost-clique that is neither small nor solitary.
    With high probability over $\rL_6(V)$, we can compute a serene coloring $\psi$ such that $C \cup \dom\col \subseteq \dom\psi$.
\end{lemma}

\begin{proof}
    We assume without loss of generality that $C$ is uncolored in $\col$, and otherwise begin by uncoloring all its vertices.
    Let $K$ be the core of $C$ and let $x_1$, $x_2$ be two $\alpha\rho^2$-friend of $C$ and $w$ their shared neighbor in $C$. 
    If $C \neq K$ and $s \in C - K$ is not $x_1$ nor $x_2$, we begin by coloring $s$. It has at least $\Delta/2$ uncolored neighbors in $C$, thus at least $\Delta/2-1 - \eps\Delta\geq\Delta/3$ available colors. The probability that none of those colors are sampled in $\rL_6(s)$ is at most $(1 - \rho^3/\Delta)^{\Delta/3} \leq 1/\poly(n)$. So we may color $s$ while keeping the coloring serene. Recall that since $C$ is not solitary, the only vertex in $C-K$ is $s$, thus only $K$ remains to color.

    We have that $x_1, x_2 \notin K$ and $w\in K$. Recall that every $v\in K$ has $a(v) \leq 1$ and $e(v) \leq \rho+1$ because $C$ is not solitary nor small, respectively. In particular, we know $N_G(v)$ for all $v\in K$. Denote by $z_1 \neq z_2 \in K$ the anti-neighbors of $x_1$ and $x_2$ respectively (whose existence is given by definition, c.f. \cref{def:popular}).
    Then $L_\col(x_i) \cap L_\col(z_i)$ contains at least $\Delta/(\alpha\rho^2) - 1 - e(z_i) \geq \Delta/(2\alpha\rho^2)$ colors, because $x_i$ has at least $\Delta/(\alpha\rho^2)$ uncolored neighbors in $K$. Since each color is sampled independently with probability $\rho^3/\Delta$, we have that
    \[
    \Prob*{ \rL_6(x_1) \cap L_\col(x_1) \cap L_\col(z_1) = \emptyset }
    \leq \paren*{ 1 - \frac{\rho^3}{\Delta} }^{\frac{\Delta}{2\alpha\rho^2}} 
    \leq \exp(-\rho/2\alpha) \leq 1/\poly(n) \ .
    \]
    So, \whp, we can serenely same-color $x_1$ and $z_1$. Similarly, \whp, we can serenely same-color $x_2$ and $z_2$ with a different color. As $w$ is incident to $x_1, x_2, z_1, z_2$, it has two units of slack. Any other vertex in $N(x_1) \cap K$ is adjacent to $w$, $x_1$ and $z_1$. So the coloring can be extended to $K$.
\end{proof}

\begin{lemma}
    \label{lem:color-solitary}
    Suppose $\col$ is a serene coloring and $C$ is a non-small solitary almost-clique. With high probability over $\rL_6(V)$, we can compute a serene coloring $\psi$ such that $C \cup \dom\col \subseteq \dom\psi$.
\end{lemma}

\begin{proof}
    We assume without loss of generality that $C$ is uncolored in $\col$, and otherwise begin by uncoloring all its vertices.
    We may also assume that $C$ contains fewer than $10^7\eps\Delta$ anti-edges, since otherwise \cref{lem:streaming-color-trouée} colors $C$ through palette sparsification. In particular, it implies that $|C| \leq \Delta+1$.
    Recall that for such $C$, we recover a solitary helper structure as specified in \cref{def:solitary-helper} (by \cref{lem:sparse-recovery}). Suppose that it is of the first kind: we know about distinct vertices $v_1, u_1, v_2, u_2 \in C$ such that $v_1u_1$ and $v_2u_2$ are disjoint anti-edges, and we know $N(u_1)$ and $N(u_2)$. We serenely same-color $v_iu_i$ with a color from $\rL_6(v_i)$. We also use different colors for both anti-edges. Since $C$ is uncolored, $v_1$ and $u_2$ have at least $\Delta/2$ shared available colors and the probability that none of them get sampled in $\rL_6(v_1)$ is 
    \[
    \Prob*{ \rL_6(v_1) \cap L_\col(v_1) \cap L_\col(u_1) = \emptyset }
    \leq \paren*{ 1 - \frac{\rho^3}{\Delta} }^{\Delta/2} 
    \leq \exp(-\rho^3/2) \leq 1/\poly(n) \ .
    \]
    The same goes for $v_2$ and $u_2$ excluding also the color picked for $v_1u_1$. Call $\chi_1$ and $\chi_2$ the color given to $v_1u_1$ and $v_2u_2$ respectively. We argue that the coloring of $C$ can be extended through palette sparsification.

    Since $C$ contains fewer than $10^7\eps\Delta \leq \Delta/10$ anti-edges, at most half of the vertices of $C$ have some anti-neighbor. Call this set $S$. Each vertex from $S$ has at least $(1 - \eps)\Delta - 4$ uncolored neighbors, hence at least $\Delta/2$ available colors. We color vertices $v\in S$ sequentially with colors from their random list $\rL_6(v)$. The probability that we do not find a color for some $v\in S$ is at most $(1-\rho^3/\Delta)^{\Delta/2} \leq 1/\poly(n)$; hence, by union bound, all vertices of $S$ get serenely colored with high probability. To color the vertices of $C - S$, observe that since they have no anti-neighbor, they also have external degree at most $\rho$. In particular, we know their neighborhood from sparse recovery. First, we color greedily the vertices of $C - S - N(v_1) \cap N(u_1) \cap N(v_2) \cap N(u_2)$. This is possible because each such vertex has at least $(1 - 5\eps)\Delta \geq 2$ uncolored neighbors in $N(v_1) \cap N(u_1) \cap N(v_2) \cap N(u_2)$. Finally, we color greedily vertices of $N(v_1) \cap N(u_1) \cap N(v_2) \cap N(u_2)$, which is possible because they have two pairs of same-colored neighbors.

    For the second type of solitary helper structure, where we recover a 3-independent set $v, u_1, u_2$ along with $N(u_1)$ and $N(u_2)$, we same-color $v, u_1$ and $u_2$ with a color from $\chi\in \rL_6(v)$. It exists \whp by the same argument as in the prior case. Then, we can color the remaining vertices like we did in the previous case.
\end{proof}
 \section{Lower bound for \texorpdfstring{$(\Delta - k)$}{(Delta-k)}-Coloring in Semi-Streaming}
\label{sec:lowerbound}

Recall that by a result of Molloy and Reed \cite{MR14}, graphs of maximum degree $\Delta$ (larger than some sufficiently large constant) that do not have certain forbidden local subgraphs can be $(\Delta - k)$-colored for all $k$ such that $(k + 1)(k + 2) \leq \Delta$. Very recently, \cite{anonymous} proved that such a $(\Delta - k)$-coloring can be computed in $\poly(\log\log n)$ rounds of \LOCAL for all such $k$ but the largest one (for which the problem requires $\Omega(n/\Delta)$ rounds).

In this section, we show that such a result cannot be obtained in semi-streaming.
\LowerBound*

\newcommand{\ALG}{\mathsf{ALG}}
\newcommand{\INDEX}{\mathsf{INDEX}}
The proof is via a reduction to the one-way communication complexity of the $\INDEX_m$ problem. In this problem, Alice is given a bit-string $x = (x_1, x_2, \ldots, x_m)$ and Bob is given an index $i\in[m]$. Alice sends one message to Bob and so that Bob correctly outputs $x_i$ with probability at least $2/3$. It is well-known that Alice needs to send $\Omega( m )$ bits to Bob (see \cite[Theorem 3.7]{KNR99} or \cite[Chapter 6]{RY_book20} for a more recent treatment).

We construct a gadget with $\Theta(\Delta)$ vertices such that Bob can infer the value of $x_i$ from a proper $c$-coloring of that graph. The graph induced by $C$ is a $(c - 3)$-clique, and the graph induced by $A \cup B$ will be the complement of that induced by $\ov{A} \cup \ov{B}$. We select $m$ pairs $e_1, e_2, \ldots, e_m$ between $A$ and $B$, but Alice includes the edge $e_j$ iff $x_j = 1$ (and the corresponding edge between $\ov{A}$ and $\ov{B}$ is added iff $x_j = 0$). Then Bob connects the endpoints of $e_i$ in $A \cup B$, and he connects all the vertices of $C$ and the endpoints of $e_i$ to the endpoints of $e_i$ in $\ov{A} \cup \ov{B}$. 
See \cref{fig:lowerbound-gadget}. 
The key observation is that the graph induced by $C$ plus the endpoints of $e_i$ in $A \cup B$ and $\overline{A} \cup \ov{B}$ form a $(c+1)$-clique minus one of the edge in $A \cup B$ or $\ov{A} \cup \ov{B}$, depending on whether $x_i = 0$ or $x_i = 1$. Since any $c$-coloring of this graph must same-color the endpoints of this missing edge, Bob recovers the value of $x_i$ from a proper $c$-coloring of this gadget.
It follows that $c$-coloring this gadget requires $\Omega( \Delta(\Delta - c + 1))$ bits. To obtain \cref{thm:lowerbound}, we create $g = \frac{m}{\Delta(\Delta - c + 1)}$ parallel gadgets, each encoding a different contiguous block of $\Delta(\Delta-c+1)$ bits from $x$. 

\begin{figure}[ht!]
    \centering
    \includegraphics[width=.9\linewidth]{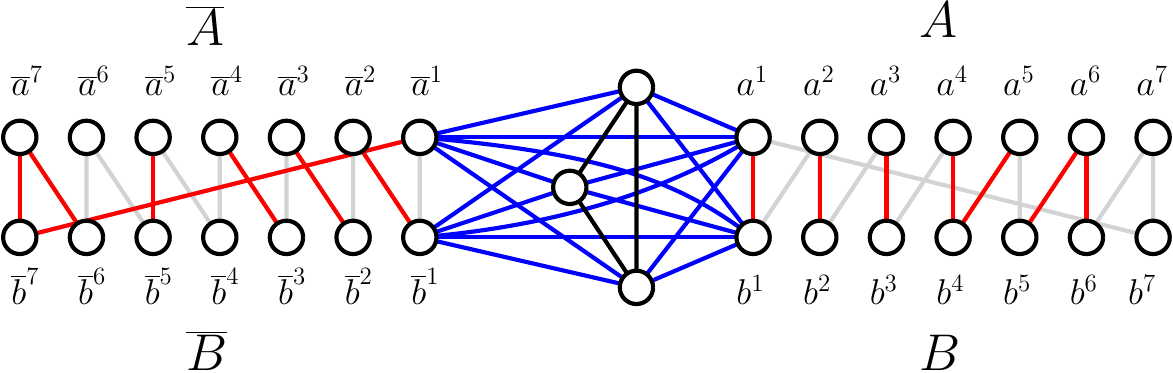}
    \caption{One gadget from the lower bound with $\Delta = 7$ and $c = 6$. In black are the edges of $C$, known to both players, and in gray the pairs used to encode the bits of $x$. The edges known/added by Alice are represented in red, and the edges known/added by Bob are represented in blue. Here $x = 10101011011000$ as read on the right with the edges between $A$ and $B$, and $i = 1$ since the edges of Bob are incident to the first gray pair. \label{fig:lowerbound-gadget}}
\end{figure}

\begin{proof}[Proof of \cref{thm:lowerbound}]
    Suppose there exists a streaming algorithm as described in \cref{thm:lowerbound}.
    Let $m$ be an integer divisible by $\Delta(\Delta - c + 1)$. Let $(x_1, x_2, \ldots, x_m) \in \set{0,1}^m$ be Alice's input and $i \in [m]$ be Bob's input to the $\INDEX$ problem. Alice and Bob construct the following graph $G(x, i)$.

    Let $g = \frac{m}{\Delta(\Delta - c + 1)}$ and $n = g(4\Delta + c - 3)$ be the number of vertices. The graph is the union of $g$ gadgets: for all $q\in[g]$, we have the following five sets of vertices
    \[
    \begin{array}{ccccc}
        A_q &= \set{ a_q^1, a_q^2, \ldots, a_q^{\Delta}}, & B_q &= \set{ b_q^1, b_q^2, \ldots, b_q^{\Delta}}, \\
        \ov{A}_q &= \set{ \ov{a}_q^1, \ov{a}_q^2, \ldots, \ov{a}_q^{\Delta}}, & \ov{B}_q &= \set{ \ov{b}_q^1, \ov{b}_q^2, \ldots, \ov{b}_q^{\Delta}}, \\
        C_q &= \set{ c_q^1, c_q^2, \ldots, c_q^{c-3}}  &
    \end{array}
    \]
    For succinctness, we call $G_q$ the subgraph of $G(x, i)$ induced by $A_q \cup B_q \cup \ov{A}_q \cup \ov{B}_q \cup C_q$ for all $q\in[g]$ (this is the gadget from \cref{fig:lowerbound-gadget}).
    An edge connects every pair of vertices in $C_q$, i.e., $C_q$ is a $(c-3)$-clique. Let $F=([2\Delta], E)$ be a $(\Delta - c + 1)$-regular bipartite graph with bipartition $\set{1, 2, \ldots, \Delta}$ and $\set{\Delta+1, \Delta+2, \ldots, 2\Delta}$. Number the edges of $F$ as $e_1, e_2, \ldots, e_t$ with $t = \Delta(\Delta - c + 1)$. 
    
    For $j \in [m]$, let $r, q \in \bN$ be such that $j = gq + r$ and $e_r = \set{ x , \Delta + y}$.
    If $x_j = 1$, Alice adds the edge $\set{a_q^x, b_q^y}$ and if $x_j = 0$ Alice adds the edge $\set{ \ov{a}_q^x, \ov{b}_q^y }$.
    Observe that the graph induced by $A_q \cup B_q$ is the complement of that induced by $\ov{A}_q \cup \ov{B}_q$, for an edge exists in one iff it does not in the other.

    Then let $i = gp + \ell$ and $e_\ell = \set{x, \Delta + y}$. Bob connects $a_p^x$, $b_p^y$, $\ov{a}_p^x$ and $\ov{b}_p^y$ to all vertices of $C_p$. He also adds and connects those four vertices as in $K_{2,2}$: he adds edges $\set{ a_p^x, \ov{a}_p^x }$, $\set{ a_p^x , \ov{b}_p^y}$, $\set{ b_p^y, \ov{a}_p^x }$, and $\set{ b_p^y, \ov{b}_p^y}$.

    Since $c\leq \Delta$, the graph $G(x, i)$ has maximum degree $\Delta$. The graph $G_q$ for $q \neq p$ is $c$-colorable since $A_q \cup B_q$ induces a bipartite graph, $C_q$ is a $(c-3)$-clique, and Bob did not add any edge between them. 
    Observe that $C_p \cup \set{a_p^x, b_p^y, \ov{a}_p^x, \ov{b}_p^y}$ is a $(c+1)$-clique minus one edge, where the edge missing is $\set{ \ov{a}_p^x , \ov{b}_p^y }$ if $x_i = 1$, and $\set{ a_p^x, b_p^y }$ if $x_i = 0$. So $C_p \cup \set{a_p^x, b_p^y, \ov{a}_p^x, \ov{b}_p^y}$ can be $c$-colored by using the same color for both endpoints of that missing edge. The vertices that are not part of this $(c+1)$-clique have degree at most $\Delta - c + 1$, thus can be greedily colored when $c > (\Delta+1)/2$. Observe that, in fact, any proper $c$-coloring of $C_p \cup \set{a_p^x, b_p^y, \ov{a}_p^x, \ov{b}_p^y}$ must same-color the endpoints of the missing edge. As such, Bob recovers the value of $x_i$ from a $c$-coloring of $G(x, i)$.

    To solve the $\INDEX$ instance $x, i$, Alice runs the streaming algorithm on all the edges of $G(x, i)$ except those added by Bob. Then she sends to Bob everything stored by the streaming algorithm, and Bob completes the streaming pass by processing his edges. Since the algorithm produces a proper $c$-coloring of $G(x,i)$ with probability at least $2/3$, Bob recovers $x_i$ with probability at least $2/3$. It follows that the streaming algorithm must use 
    \[ 
    \Omega(m) = \Omega(g\Delta(\Delta - c + 1)) \geq \Omega( n (\Delta - c + 1)) 
    \]
    space, where the last inequality uses that $g \geq n/(5\Delta)$ since $c \leq \Delta$.
\end{proof}

\printbibliography
\appendix
\section{Sparse Recovery}
\label{sec:sparse-recovery-details}

Our sparse recovery technique follows closely that of \cite{AKM23}. The first ingredient is a property of the Vandermonde matrix that allows to recover sparse vectors. A vector $x \in \bF^n$ is \underline{$k$-sparse} if it has at most $k$ non-zero coordinates. We use the following well-known fact:
\begin{proposition}
    \label{lem:vandermonde}
    For $n, k \geq 1$ arbitrary integers and $p \geq n$ a prime number, the Vandermonde matrix $\Phi^V$ of dimension $(2k \times n)$ over $\bF_p$ is defined as $\Phi^V_{i, j} = j^{i-1} \mod p$ for all $i \in [2k]$ and $j\in[n]$. For any $k$-sparse vector $x\in \bF_p$, one can uniquely recover $x$ from $\Phi^V x$ in polynomial time.
\end{proposition}

Since we do not know whether vectors are $k$-sparse or not, we test the outcome produced by the recovery algorithm of \cref{lem:vandermonde}. We use random matrices; see \cite[Proposition 3.7]{AKM23} for a proof.
\begin{proposition}
    \label{lem:random-matrix}
    Let $n, t \geq 1$ be arbitrary integers and $p \geq n$ a prime number. If $x \neq y\in \bF_p^n$ and $\Phi^R \in \bF_p^{(t \times n)}$ is a random matrix, we have 
    \[
    \Prob*{ \Phi^R x = \Phi^R y } \leq p^{-t} \ .
    \]
\end{proposition}

The streaming algorithm is the same as that of \cite{AKM23} up to minor changes in the parameters. The recovery algorithm, after the streaming pass, is essentially the same except that we recover slightly more information in solitary almost-cliques.

Recall that $\rho = \Theta(\eps^{-2}\log n)$. For a set $S \subseteq [n]$, denote by $\boldsymbol{1}_S$ the characteristic vector of $S$.

\begin{Algorithm}
    \label{alg:sparse-recovery}
    Sparse Recovery.
    \medskip

    Let $p$ be any prime number between $n^c$ and $2n^c$.

    For every $i = 0, 1, 2, \ldots, \ceil{\log\paren*{ \Delta / \rho } }$, let $s_i = 2^i \rho$ and do:
    \begin{enumerate}
        \item Let $\Phi_i^V \in \bF_p^{2s_i \times n}$ be the Vandermonde matrix and $\Phi^R_i \in \bF_p^{c \times n}$ be a random matrix.
        \item Sample each $v\in V_i$ independently with probability $p_i = \min\set{1, \frac{4\rho}{s_i}}$. For every vertex $v \in V_i$, initialize two vectors $y_i(v) = \Phi^V_i \boldsymbol{1}_v \in \bF_p^{2s_i}$ and $z_i(v) = \Phi^R_i \boldsymbol{1}_v \in \bF_p^{c}$.
        \item For each $v \in V_i$, when an edge $\set{u,v}$ appears in the stream, do
        \[
            y_i(v) \gets y_i(v) + \Phi^V_i \boldsymbol{1}_u
            \quad\text{and}\quad
            z_i(v) \gets z_i(v) + \Phi^R_i \boldsymbol{1}_u \ .
        \]
    \end{enumerate}
\end{Algorithm}

\begin{proof}[Proof of \cref{lem:sparse-recovery}]
    Let us first argue about the space usage of \cref{alg:sparse-recovery}. For a given $i$ and vertex $v\in V_i$, the vectors $y_i(v) \in \bF_p^{2s_i}$ and $z_i(v) \in \bF_p^c$ require $O(s_i\log p)$ bits of memory to store. When $i \in \set{0,1,2}$, we have $V = V_i$, so the algorithm uses $O(ns_1\log n) = O(n\log^2 n)$ bits of memory. Then, each $V_i$ for $i \geq 3$ has expected size $n \cdot 4\rho/s_i$, thus by the Chernoff Bound, 
    \begin{align*}
    \Prob*{ |V_i| > 8n\rho/s_i } \leq \exp\paren*{ -\frac{4n\rho}{3s_i} }
    \leq 1/\poly(n) \ ,\\
    \tag{using that $s_i \leq 2^{\log(\Delta/\rho)+1}\rho \leq 2\Delta \leq 2n$}
    \end{align*}
    resulting in $\sum_i O( n\rho\log p ) =  O(n\log^3 n)$ bits of memory for storing all the $y_i(v)$ and $z_i(v)$ vectors. Storing all the random matrices $\Phi^R_i$ requires $O(n\log p \log n) = O(n\log^2 n)$ memory. It is not necessary to store $\Phi_i^V$ explicitly for its coefficients can be derived from $s_i$, $p$ and $n$. Overall, we use $O(n\log^3 n)= O(n\log^3 n)$ bits of memory.

    Now, consider $v \in C$. Suppose that $a(v) + e(v) \leq 4\rho$. Observe that $a(v)$ and $e(v)$ are unknown at this point. At the end of the streaming pass, for every $v\in V$, we have that $y_2(v) = \Phi^V_2 \boldsymbol{1}_{N[v]}$, so the vector $x(v) = \Phi_2^V \boldsymbol{1}_C - y_2(v)$ is equal to $+1$ at coordinates $u\in A(v)$, $-1$ for $u\in E(v)$ and zero otherwise. By assumption $x(v)$ is $4\rho$-sparse and we recover $A(v)$ and $E(v)$, by \cref{lem:vandermonde}; hence, we recover $N(v)$. Now, suppose that $a(v) + e(v) > 4\rho$ so that $x(v)$ is not $\rho$-sparse, but by using the algorithm of \cref{lem:vandermonde} on $\Phi^V \cdot x(v)$, we recover some set $S \subseteq V$. To test if $S = N(v)$, we compute $\Phi_2^R \boldsymbol{1}_S$ and test if it equals $z_2(v) = \Phi_2^R \boldsymbol{1}_{N(v)}$. By \cref{lem:random-matrix}, if $S \neq N(v)$, we obtain different values with high probability.

    Consider a non-small solitary almost-clique $C$. Observe that if $\set{u,v}$ is an anti-edge of $C$ such that $a(u) + e(v) \leq 4\rho$, then by the earlier argument we recover $N(u)$ and learn about the anti-edge $\set{u,v}$. So we may assume that all anti-edges in $G[C]$ have both endpoints with $a(v) + e(v) > 4\rho$. Observe it suffices to construct the 2-anti-matching helper structure from \cref{def:solitary-helper}, because when there is a 3-independent set and no 2-anti-matching, we recovered the neighborhood of all vertices.

    Pick the two vertices $v_0 \neq v_1$ with largest anti-degrees in $C$. Observe that since $C$ is not small, for both $i\in \set{0,1}$, we have that
    \begin{equation}
        \label{eq:lowerbound-anti}
    4\rho \leq a(v_i) + e(v_i) \leq \Delta + 1 - |C| + 2a(v_i)
    \leq \rho + 2a(v_i) \ ,
    \end{equation}
    and hence that $a(v_i) \geq \rho$.
    Fix $v_i$ and call $j \geq 1$ the integer for which $s_j  \leq a(v_i) \leq s_{j+1}$ and $A = A(v_i) - v_{1-i}$. By the same reasoning as for \cref{eq:lowerbound-anti}, every $u\in A$ has 
    \[
    a(u) + e(u) 
    \leq \rho + 2a(u)
    \leq \rho + 2s_{j+1} \leq s_{j+3} \ .
    \]
    In particular, the vector $x(u) = \Phi_{j+3}^V \boldsymbol{1}_C - y_{j+3}(u)$ is $s_{j+3}$-sparse for all $u\in A$. The probability that none of the $u \in A$ gets sampled in $V_{j+3}$ is at most
    \[
    (1 - p_{j+3})^{a(v_i) - 1} 
    \leq \exp\paren*{ - \frac{\rho}{s_{j+3}} (s_j - 1) } 
    \leq \exp(-\rho/16) \leq 1/\poly(n) \ ,
    \]
    where the last inequality holds because $s_j - 1 = s_{j+2}/4 - 1 \geq s_{j+2}/8 = s_{j+3}/16$ (using $s_{j+2} \geq 8$) and $\rho \geq \Theta(\log n)$. After the streaming pass, we recover $N(u)$ for every $s_i$-sparse vertex of $V_i$ --- eliminating those that are not sufficiently sparse with the test of \cref{lem:random-matrix}. For $v_0$ and $v_1$ in $C$, we thereby identify a 2-anti-matching $\set{v_0u_0, v_1u_1}$ where we recovered $N(u_0)$ and $N(u_1)$.
\end{proof}

 \section{Degree-Minus-One Choosable Subgraphs}
\label{sec:choosable-subgraphs}

A graph $Q=(V,E)$ is \underline{$(\deg-1)$-choosable} if, for any lists $L : V \to 2^{\mathbb{N}}$ such that $|L(v)| = \deg(v, Q)-1$ for all $v\in V$, there exists a coloring $\col$ of $Q$ such that $\col(v) \in L(v)$ for all $v\in V$.

The following lemma implies that every almost-clique with a 2-anti-matching or any popular almost-clique induces a $(\deg-1)$-choosable subgraph of size 8.

\begin{lemma}
    Fix $d \geq 5$.
    Let $Q$ be a graph on vertex set $V = K + s_1 + s_2$ where $Q[K]$ is a $d$-clique and 
    \begin{enumerate}
        \item there exists some $v\in N(s_1) \cap N(s_2) \cap K$,
        \item both $s_1$ and $s_2$ are incident to at least $5$ vertices in $K$, and
        \item there are two vertices $a_1 \neq a_2 \in K$ such that $a_1 \notin N(s_1)$ and $a_2 \notin N(s_2)$.
    \end{enumerate}
    Then $Q$ is $(\deg-1)$-choosable.
    This holds regardless of the existence of the edge $\set{s_1, s_2}$.
    \label{lem:d-1CS}
\end{lemma}

\begin{proof}
    Fix arbitrary lists $L : V \to 2^{\bN}$ such that $|L(u)| = \deg(u)-1$ for every $u\in V$. More precisely, $L(v)$ contains exactly $d$ colors, $L(s_1)$ and $L(s_2)$ contain between $4$ and $d$ colors, and every $u\in K$ has at least $d-2$ colors in its list.

    We show by case analysis that we can provide two units of slack to $v$ and one unit of slack to another uncolored vertex. It implies that the coloring of $Q$ can be completed greedily.

    \underline{Case 1:}
    For at least one $i \in \set{1,2}$, we have $L(s_i) \cap L(a_i) \neq \emptyset$. Let us assume without loss of generality that $i=1$.
    Same-color $\set{s_1, a_1}$ with $\chi_1 \in L(s_1) \cap L(a_1)$. This provides $v$ with at least one unit of slack and some uncolored neighbor of $s_1$ in $K - v - a_2$ with unit slack (regardless of whether its list contains $\chi_1$ or not). We may assume that $\chi_1 \in L(v)$ as otherwise, $v$ has two units of slack, and we are done already. If there exists $\chi_2 \in L(s_2) \cap L(a_2) - \chi_1$, we same-color $s_2$ and $a_2$ with $\chi_2$ which provides $v$ a second unit of slack. So we henceforth assume $L(s_2) \cap L(a_2) - \chi_1$ is empty. But in that case, there exists a color $\chi_2 \in L(s_2) \cup L(a_2) - \chi_1$ that is not in $L(v) $ because otherwise
    \[
    d-1 = |L(v) - \chi_1|
    \geq |L(s_2) - \chi_1| + |L(a_2) - \chi_1|
    \geq (4-1) + (d-3) = d \ .
    \]
    Coloring $s_2$ or $a_2$ with such a $\chi_2\notin L(v) - \chi_1$ provides $v$ its second unit of slack.

    \underline{Case 2:}
    For both $i\in\set{1,2}$ sets $L(s_i)$ and $L(a_i)$ are disjoint. By the same argument as in Case 1, there must exist a color $\chi_1\in L(s_1) \cup L(a_1)$ that is not in $L(v)$. Color $s_1$ or $a_1$ with $\chi_1$ and leave the other vertex uncolored. This provides one unit of slack to $v$. Then, there must be $x\in \set{s_2, a_2}$ such that $L(x) - (L(v) + \chi_1)$ contains at least one color. We call $y \in \set{s_2, a_2} - x$ the other vertex of the pair. Such an $x$ exists, because $L(x)$ and $L(y)$ are disjoint and if $L(v)$ contains both, then
    \[
    d = |L(v)|
    \geq |L(x) - \chi_1| + |L(y) - \chi_1|
    \geq (d-2) + 4 - 1 = d+1 \ ,
    \]
    a contradiction. We subtract $\chi_1$ (used by $s_1$ or $a_1$) only once as it can be in at most one of the two lists. There are now two subcases: let $w$ be a shared neighbor of $x$ and $y$ in $K - v - a_1$ (which exists because of (2))
    \begin{itemize}
        \item If there exists $\chi_2 \in (L(x) - (L(v)  + \chi_1)) - L(w)$ then coloring $x$ with $\chi_2$ provides both $v$ and $w$ unit slack (and we are done).
        \item Otherwise, $L(x) - (L(v) + \chi_1) \subseteq L(w)$ and coloring $w$ with any color $\chi_2 \in L(x) - (L(v) + \chi_1)$ provides one unit of slack to both $v$ and $y$. Vertex $y$ receives unit slack because $\chi_2 \in L(x)$ which is disjoint from $L(y)$.\qedhere
    \end{itemize}
\end{proof}

For the remaining type of solitary almost-cliques, we have the following result. The argument is similar in nature but the case analysis differs. Here again, the resulting subgraph contains 9 vertices.

\begin{lemma}
    Fix $d \geq 6$. Let $Q$ be a graph on vertex set $V = K + u_1 + u_2 + u_3$ where $Q[K]$ is a $d$-clique and 
    \begin{enumerate}
        \item $\set{u_1, u_2, u_3}$ is an independent set,
        \item each $u_i$ is adjacent to all vertices of $K$.
    \end{enumerate}
    Then $Q$ is a $(\deg-1)$-choosable.
\end{lemma}

\begin{proof}
    Fix lists $L : V \to 2^{\mathbb{N}}$ such that $|L(v)| = \deg(v)-1$ for all $v$. Each $v\in K$ has a list of $d+1$ colors and each $u_i$ has a list of $d-1$ colors.

    Suppose first that there exists a pair $u_i \neq u_j$ for which $L(u_i) \cap L(u_j)$ contains at least 3 colors. Call $v$ any vertex of $K$ and $u_k \neq u_i, u_j$ the remaining $u$-vertex. If $L(u_k)$ shares a color with $L(u_i) \cap L(u_j)$, then we can extend the coloring by same-coloring all three vertices. If $L(u_i) \cap L(u_j)$ is not included in $L(v)$, then we can extend the coloring by same-coloring $u_iu_j$ with a color outside $L(v)$. If $L(u_i) \cap L(u_j) \subseteq L(v)$ but $L(u_k)$ contains a color outside $L(v)$, then we same-color $u_iu_j$ and color $u_k$ with a color not in $L(v)$. This is all the cases because otherwise $L(u_i) \cap L(u_j)$ and $L(u_k)$ are disjoint subsets of $L(v)$, which implies the following absurdity
    \[
    d+2= 3 + (d-1) \leq |L(u_i) \cap L(u_j)| + |L(u_k)| \leq |L(v)| = d+1 \ .
    \]

    Suppose now that all pairs $i\neq j$ are such that $|L(u_i) \cap L(v_j)| \leq 2$. Consider first $u_1$ and $u_2$. Note that $L(u_1) \cup L(u_2)$ contains at least $2(d-3) + 2 = 2d-4$ colors. For $d \geq 6$, we can color either $u_1$ or $u_2$ with a color not in $L(v)$. Suppose without loss of generality that this vertex is $u_1$. By the same reasoning, we color $u_2$ or $u_3$ with a color not in $L(v)$. The coloring can then be extended to all vertices.
\end{proof}
 
\end{document}